\documentclass[a4paper,pdftex,reqno,10pt]{amsart}%

\usepackage{amssymb,amsmath}
\usepackage{mathptmx}       
\usepackage{helvet}         
\usepackage{courier}        
\usepackage{type1cm}        
\usepackage{url}
\usepackage{comment}  
\usepackage{algorithm2e}
\usepackage{graphicx}        
\usepackage{multicol}        
\usepackage[bottom]{footmisc}

\usepackage{bm} 
\usepackage{enumerate} 
\usepackage{ytableau} 

\newcommand{\neff}{n_{\operatorname{eff}}}

\newcommand{\code}[4]{[#1,#2,#3]_{#4}}
\newcommand{\codeshort}[3]{[#1,#2]_{#3}}
\newcommand{\codeeff}[4]{\left[\underline{#1},#2,#3\right]_{#4}}
\newcommand{\codeshorteff}[3]{\left[\underline{#1},#2\right]_{#3}}

\newcommand{\snumb}[3]{s_{#3}(#1,#2)} 

\newtheorem{theorem}{Theorem}
\newtheorem{proposition}{Proposition}
\newtheorem{lemma}{Lemma}

\makeindex             
\begin{document}

\title[Classification of $8$-divisible binary linear codes with minimum distance $24$]{Classification of $\mathbf{8}$-divisible binary linear codes with minimum distance $\mathbf{24}$}

\author{Sascha Kurz}
\address{Sascha Kurz, University of Bayreuth, 95440 Bayreuth, Germany}
\email{sascha.kurz@uni-bayreuth.de}


\abstract{We classify $8$-divisible binary linear codes with minimum distance $24$ and small length. As an application we consider 
the codes associated to nodal sextics with $65$ ordinary double points.\\[2mm]
\textbf{Keywords:} triply even codes, divisible codes, classification, nodal sextics\\
\textbf{MSC:} Primary  94B05.
}}
\maketitle

\section{Introduction}
Doubly even codes were subject to extensive research in the last years. For applications and enumeration results 
we refer e.g.\ to \cite{doran2011codes}. More recently, triply even codes were studied, see e.g.\ \cite{betsumiya2012triply,honold2019lengths}.  
These two classes of binary linear codes are special cases of so-called $\Delta$-divisible codes, where all weights are divisible 
by $\Delta$. Being introduced by Ward, see \cite{ward2001divisible} for a survey, they have many applications. A recent example 
is the maximum size of partial spreads, i.e., sets of $k$-dimensional subspaces of $\mathbb{F}_q^v$ with trivial intersection and 
maximum possible cardinality. All currently known upper bounds for partial spreads can be deduced from non-existence results for 
$q^{k-1}$-divisible projective codes, see \cite{heinlein2019projective,honold2018partial}. For some enumeration results for projective $2^r$-divisible 
codes we refer to \cite{ubt_eref40887}. It has been observed in \cite{honold2018partial} that among the linear codes with maximum 
possible minimum distance $d$ there are often examples which are $q^r$-divisible, provided that $q^r$ divides $d$. Here we study the 
special case of triply even, i.e., $8$-divisible binary  linear codes with minimum distance $d=24$. We exhaustively enumerate all 
such codes for small lengths. While those classification results are of cause of interest in coding theory, there is another motivation 
coming from algebraic geometry. A \textit{nodal surface} is a hypersurface of degree $s$ in $\mathbb{P}_3(\mathbb{C})$ with $\mu$ ordinary 
double points (nodes) as its only singularities. The maximum number $\mu(s)$ of nodes was determined by Cayley \cite{cayley1869vii} and 
Schl\"afli \cite{dr1863distribution} for $s=3$ and by Kummer \cite{kummer1864flachen} for $s=4$, respectively. In \cite{beauville1979nombre} 
Beauville concluded the existence of a binary linear code $C$ in $\mathbb{F}_2^n$ with certain further properties from the existence of a 
nodal surface with $m\ge n$ nodes. This connection  allowed him to overcome the general upper bound of Basset \cite{basset1906maximum} and 
especially to determine $\mu(5)=31$. The coding theoretic approach was used in \cite{jaffe1997sextic} to obtain $\mu(6)<66$, so that 
$\mu(6)=65$ due to the existence of the so-called Barth sextic \cite{barth1996two}. In \cite[Theorem 5.5.9]{PhdPettersen} a unique irreducible 
$3$-parameter family of $65$-nodal sextics containing the Barth sextic was determined. For the next case only $99\le \mu(7)\le 104$ is known, 
see \cite{labs2004septic} and \cite{varchenko1983semicontinuity}, respectively. The following general properties of the associated code $C$ of a nodal 
surface with degree $s$ and $m$ nodes are known. For the dimension $k$ of $C$ a general argument of Beauville \cite{beauville1979nombre} 
gives $k\ge m-\left\lceil s^3/2\right\rceil +2s^2-3s+1$, see \cite[Proposition 4.3]{jaffe1997sextic}. If $s$ is odd, then $C$ is doubly even and triply even 
otherwise, see \cite[Proposition 2.11]{catanese1981babbage}. The minimum distance $d$ satisfies $d\ge 2\lceil s(s-2)/2\rceil$, see 
\cite[Theorem 1.10]{endrass1997minimal}. In some cases further weights can be excluded. For a more extensive overview on the history and technical 
details of nodal surfaces with many nodes we refer the interested reader e.g.\ to \cite{PhdLabs}. 

The remaining part of the paper is organized as follows. In Section~\ref{sec_computer_classification} we describe algorithms for the exhaustive generation 
of linear codes and apply them for $8$-divisible binary linear codes with minimum distance $24$ and small parameters. As an application codes of nodal 
surfaces are considered in Section~\ref{sec_nodal_surfaces}. In Section~\ref{sec_theoretical} we collect some theoretical arguments that partially 
show our computational findings from the two previous sections. 
The optimal $\code{n}{k}{24}{2}$ codes that are $8$-divisible are tabulated in Appendix~\ref{sec_class_optimal}.  

\section{Computer classification of linear codes}
\label{sec_computer_classification}

A \emph{$q$-ary linear code} $C$ of \emph{length} $n$ and \emph{dimension} $k$, or an $\codeshort{n}{k}{q}$ code, is a $k$-dimensional subspace of $\mathbb{F}_q^n$. 
It can be represented by a basis. Written rowwise this is called \emph{generator matrix} in coding theory. An example of a $\codeshort{64}{13}{2}$ code 
is given by: 
$$
\begin{pmatrix}
1111111111111111111111111111111111111111111111111111111111111111\\
0000000000000000000000000000000011111111111111111111111111111111\\
0000000000000000111111111111111100000000000000001111111111111111\\
0000000011111111000000001111111100000000111111110000000011111111\\
0000111100001111000011110000111100001111000011110000111100001111\\
0011001100110011001100110011001100110011001100110011001100110011\\
0101010101010101010101010101010101010101010101010101010101010101\\
0000000000001111001100110011110000111100001100111111000011111111\\
0001001000011101000100100001110101001000010001110100100001000111\\
0000011001100000001110101010001101011100110001011001111111111001\\
0000000001101001000000001001011001011010110011000101101000110011\\
0001001000011101011110110111010000101110110111101011100001001000\\
0000001101010110011001011100111100000011101010011001101011001111\\
\end{pmatrix}
$$
Elements of such a subspace are called \emph{codewords}. The \emph{weight} of a codeword is the number of non-zero coordinates. So, each non-empty 
linear code contains exactly one codeword of weight zero. The \emph{minimum distance} $d$ of a linear code is the smallest non-zero weight of a 
codeword. If $a_i$ denotes the number of codewords of weight $i$, then the \emph{weight enumerator} is given by 
$W(z)=\sum_{i\ge 0} a_iz^i$. In our example we have
$$
  W(z)=1z^{0}+1008z^{24}+6174z^{32}+1008z^{40}+1z^{64},
$$
i.e., the minimum distance is given by $d=24$. Adding zero columns to the above generator matrix does neither change the dimension, the minimum 
distance nor the weight enumerator. However, the length is increased so that we call the smallest possible length the \emph{effective length} $\neff$. To ease 
the notation we write $\codeshorteff{n}{k}{q}$ for a $k$-dimensional code over $\mathbb{F}_q$ with effective length $\neff=n$. If we want to highlight the minimum  
distance $d$ of a code we speak of an $\code{n}{k}{d}{q}$ or $\codeeff{n}{k}{d}{q}$ code. If only weights from a set 
$\{w_1,\dots,w_l\}\subset\mathbb{N}$ can occur in the code we speak of an $\code{n}{k}{\{w_1,\dots,w_l\}}{q}$ code. We also use notations as $\codeshorteff{\le n}{k}{q}$ and 
$\codeshorteff{n}{\ge k}{q}$, as well as their variants, in order to denote the set of all $q$-ary $k$-dimensional linear codes with effective length at most $n$ and the set 
of all $q$-ary linear codes with effective length $n$ and dimension at least $k$, respectively. The dual of an $\codeshort{n}{k}{q}$ code $C$ is the set of all codewords in $\mathbb{F}_q^n$ 
that are perpendicular on $C$. The \emph{dual code} $C^*$ has length $n$ and dimension $n-k$. By $a_i^*$ we denote the number of codewords of weight $i$ of the 
dual code, so that we can also speak of the minimal dual distance $d^*$. In our example the minimal dual distance is $4$. The weight distribution of a linear code 
and its dual is related by the so-called MacWilliams identities, see e.g.\ \cite{macwilliams1977theory}. For an $\codeshorteff{n}{k}{2}$ code we have:
\begin{equation}
  \label{eq_mac_williams}
  \sum a_i^* x^{n-i}y^i=\tfrac{1}{2^k}\cdot\sum a_i(x+y)^{n-i}(x-y)^i. 
\end{equation}  

Given a linear code $C$ we can consider the span of the columns of a generator matrix of $C$, i.e., we have an associated multiset of $1$-dimensional subspaces, called 
points, of $\mathbb{F}_q^n$. Starting from a multiset of points we can naturally associate a code, see e.g.\ \cite{dodunekov1998codes} for more details. 
Geometrically $d^*\ge 3$ means that the associated multiset of points is indeed a set of points, i.e., the code is \emph{projective}. $d^*\ge 4$ translates to the 
geometrical fact that the associated set of points  does not contain a full line. Permuting columns 
of a generator matrix of a linear code does not change the key parameters of the code and is considered as the set of automorphisms. Here we restrict ourselves 
to the automorphisms of the corresponding multiset of points which ignores permutations of identical columns. The automorphism group of our example has order 
$\# \operatorname{Aut}=23224320$. The code was obtained in \cite{delsarte1975alternating} and has the following nice description, see \cite{jaffe1997sextic}:
It is a subcode of the second order Reed-Muller code $R(2,6)$ containing the first order Reed-Muller code $R(1,6)$ as a subcode. The
cosets of $R(1,6)$ in it correspond to the symplectic forms $B_a$ in $\mathbb{F}_{64}$, given by $B_a(x,y)=\operatorname{tr}((ax^4 + a^{16}x^{16})y)$.

One way to generate linear $\code{n}{k}{W}{q}$ codes with weights in some set $W\subseteq \mathbb{N}$ is to start from an $\code{n'}{k-1}{W}{q}$ subcode, where 
$n'\le n-1$, and to append another row to the generator matrix. This approach consists of two steps. First one has to determine candidates for the additional row 
of the generator matrix that lead to an $\codeshort{n}{k}{q}$ code with weights in $W$ and then one has to filter out the non-isomorphic copies, c.f.\ \cite{bouyukliev2007q}. 
We start by formulating the first part as an enumeration problem of integral points in a polyhedron:
\begin{lemma}
  \label{lemma_ILP}
  Let $G$ be a systematic generator matrix of an $\codeshorteff{n}{k}{2}$ code whose weights are $\Delta$-divisible and are contained in $\left[a\cdot \Delta,b\cdot \Delta\right]$. 
  By $c(u)$ we denote the number of columns of $G$ that equal $u$ for all $u$ in $\mathbb{F}_2^k\backslash\mathbf{0}$, $c(\mathbf{0})=n'-n$,  and let $\mathcal{S}(G)$ be 
  the set of feasible solutions of 
  \begin{eqnarray}
    \Delta y_h+\sum_{v\in\mathbb{F}_2^{k+1}\,:\,v^\top h=0} x_v =n-a\Delta&&\forall h\in\mathbb{F}_2^{k+1}\backslash\mathbf{0}\label{eq_hyperplane}\\
    x_{(u,0)}+x_{(u,1)} =c(u) && \forall u\in\mathbb{F}_2^{k} \label{eq_c_sum}\\
    x_{e_i}\ge 1&&\forall 1\le i\le k+1\label{eq_systematic}\\
    x_v\in \mathbb N &&\forall v\in\mathbb{F}_2^{k+1}\\ 
    y_h\in\{0,...,b-a\} && \forall h\in\mathbb{F}_2^{k+1}\backslash\mathbf{0}\label{hyperplane_var},
  \end{eqnarray}
  where $e_i$ denotes the $i$th unit vector in $\mathbb{F}_2^{k+1}$ and $n'\ge n+1$. Then, for every systematic generator matrix $G'$ of an $\codeshorteff{n'}{k+1}{2}$ code $C'$ 
  whose first $k$ rows coincide with $G$ we have a solution $(x,y)\in\mathcal{S}(G)$ such that $G'$ has exactly $x_v$ columns equal to $v$ for each $v\in\mathbb{F}_2^{k+1}$.
\end{lemma}
\begin{proof}
  Let such a systematic generator matrix $G'$ be given and $x_v$ denote the number of columns of $G'$ that equal $v$ for all $v\in\mathbb{F}_2^{k+1}$. Since $G'$ is systematic, 
  Equation~(\ref{eq_systematic}) is satisfied. As $G'$ arises by appending a row to $G$, also Equation~(\ref{eq_c_sum}) is satisfied. Obviously, the $x_v$ are non-negative 
  integers. The conditions (\ref{eq_hyperplane}) and (\ref{hyperplane_var}) correspond to the restriction that the weights are $\Delta$-divisible and contained in 
  $\{a\Delta,\dots,b\Delta\}$.   
\end{proof}
We remark that also every solution in $\mathcal{S}(G)$ corresponds to an $\codeshorteff{n'}{k+1}{2}$ code $C'$ with generator matrix $G'$ containing $C$ as a subcode. The method can also be 
easily adopted to field sizes $q>2$ by simply counting $1$-dimensional subspaces in $C$ and $x$ instead of vectors. Half of the constraints (\ref{eq_hyperplane}) 
are automatically satisfied since $C$ satisfies all constraints on the weights. If there are further forbidden weights in $\{i\Delta\,:a\le i\le b\}$ then, one may 
also use the approach of Lemma~\ref{lemma_ILP}, but has to filter out the integer solutions that correspond to codes with forbidden weights. Another application of this 
first generate, then filter strategy is to remove some of the constraints (\ref{eq_hyperplane}), which speeds up, at least some, lattice point enumeration algorithms.  

For the first part, i.e., the application of Lemma~\ref{lemma_ILP}, we use an implementation of the LLL lattice point enumeration algorithm, see \cite{wassermann2002attacking}. 
For the filtering of non-isomorphic copies we have used the software \texttt{Q-Extension} \cite{bouyukliev2007q} or \texttt{CodeCan} \cite{feulner2009automorphism}.
It remains to specify the choice of the parameters $n$, $n'$, and $k$. In order to generate $\codeshorteff{n'}{k+1}{2}$ codes all $\codeshorteff{n}{k}{2}$ codes with $n<n'$ have  
to be known, so that the generation is performed with increasing dimension $k$. However, this way we get a lot of isomorphic copies since a $\codeshorteff{[n'}{k+1}{2}$ code $C'$ 
usually contains several non-isomorphic $\codeshorteff{n}{k}{2}$ subcodes $C$. To slightly reduce this effect, we assume that every column of the generator matrix of $C$ is 
contained at least $n'-n$ times, since otherwise there exists a $\codeshorteff{\hat{n}}{k}{2}$ code $\hat{C}$ with $\hat{n}>n$ that can be extended to $C'$. In other words, we 
assume that the vector of the effective lengths in the generation path of a code is weakly decreasing. We remark that more sophisticated assumptions on the order of the generation 
of subcodes can be made to even better overcome the problem of the generation of a huge number of isomorphic codes. However, in order to be even resistant to a some local 
hardware failures in our computations, we have decided not to implement those. 
  
We have cross checked\footnote{The $\codeeff{\le 60}{7}{\{24,32,40\}}{2}$ codes have also been generated by solely using \texttt{Q-Extension}. 
As the $\codeeff{n}{k}{\{24,32,40,48,56,64\}}{2}$ codes contain the $\codeeff{n}{k}{\{24,32,40\}}{2}$ codes, we have another 
cross check.} our algorithms and implementations with the case of $4$-divisible codes treated by Miller et al. \cite{doran2011codes}, 
\url{https://rlmill.github.io/de_codes}. For all such codes with $n\le 28$ and $k\le 7$ our numbers coincide. Note that there are $1452663$ $4$-divisible 
$\codeshorteff{28}{7}{2}$ codes. In the meantime the algorithmic approach described above is implemented in more generality, see \cite{kurz2019lincode} for the details.

We remark that other approaches for classifying linear codes can e.g.\ be found in \cite[Section 7.3]{kaski2006classification} or 
\cite{betten2006error,bouyukliev2007q,feulner2014classification}.

In tables (\ref{tab_8div_p1})-(\ref{tab_8div_p3}) we have stated the number of $8$-divisible $\codeshorteff{n}{k}{2}$ codes with minimum distance $24$, dimension 
$k\le 13$, and small lengths. Note that blank entries on the left of each row correspond to a zero, while blank entries on the right of each 
row correspond to values that are not computed due to the exponential growth of the number of codes.  

\begin{table}[htp]
  \begin{center}
    \begin{tabular}{rrrrrrrrrrrrrrrrr}
      \hline
      k/n & 24 & 32 & 36 & 40 & 42 & 44 & 45 & 46 & 47 & 48 & 49 & 50 & 51 & 52 & 53 &  54 \\ 
      \hline
      1   &  1 &  1 &  0 &  1 &  0 &  0 &  0 &  0 &  0 &  1 &  0 &  0 &  0 &  0 &  0 &   0 \\ 
      2   &    &    &  1 &  1 &  0 &  2 &  0 &  0 &  0 &  3 &  0 &  0 &  0 &  3 &  0 &   0 \\
      3   &    &    &    &    &  1 &  1 &  0 &  2 &  0 &  4 &  0 &  3 &  0 &  6 &  0 &   8 \\ 
      4   &    &    &    &    &    &    &  1 &  1 &  2 &  4 &  1 &  4 &  5 & 15 &  5 &  23 \\
      5   &    &    &    &    &    &    &    &    &  1 &  4 &  1 &  6 &  5 & 30 & 15 &  92 \\
      6   &    &    &    &    &    &    &    &    &    &  1 &  1 &  2 &  5 & 21 & 29 & 160 \\
      7   &    &    &    &    &    &    &    &    &    &    &    &  1 &  1 &  4 &  7 &  58 \\
      8   &    &    &    &    &    &    &    &    &    &    &    &    &  1 &  0 &  0 &   1 \\                    
      \hline 
    \end{tabular}
    \caption{Number of $8$-divisible $\codeshorteff{n}{k}{2}$ codes with minimum distance $24$ -- part 1.}
    \label{tab_8div_p1}
  \end{center}
\end{table}

\begin{table}[htp]
  \begin{center}
    \begin{tabular}{rrrrrrrrrrrrrrrrr}
      \hline
      k/n &  55 &   56 &   57 &    58 &    59 &     60 &      61 &      62 \\ 
      \hline
      1   &   0 &    1 &    0 &     0 &     0 &      0 &       0 &       0 \\
      2   &   0 &    4 &    0 &     0 &     0 &      5 &       0 &       0 \\
      3   &   0 &   15 &    0 &    10 &     0 &     23 &       0 &      21 \\ 
      4   &  19 &   68 &   13 &    78 &    40 &    201 &      41 &     259 \\
      5   &  88 &  411 &  180 &   992 &   687 &   3384 &    1478 &    8040 \\
      6   & 303 & 1813 & 2026 & 11696 & 14870 &  83368 &         &         \\
      7   & 143 & 1493 & 3604 & 34945 & 93503 & 852947 &         &         \\
      8   &   4 &   55 &   61 &  1486 & 10971 & 376697 & 1900541 &         \\
      9   &     &    2 &    0 &     4 &    14 &    618 &   19362 & 2410702 \\
     10   &     &      &      &       &       &      6 &       8 &     682 \\
     11   &     &      &      &       &       &        &         &       3 \\                      
      \hline 
    \end{tabular}
    \caption{Number of $8$-divisible $\codeshorteff{n}{k}{2}$ codes with minimum distance $24$ -- part 2.}
    \label{tab_8div_p2}
  \end{center}
\end{table}

The computations were performed on a linux cluster of the university of Bayreuth set up in 2009. This elderly 
computing cluster consists of roughly 250 nodes with Intel Xeon E5 processors with 8 physical cores, 2.3 gigacycles, 
and 24 gigabyte RAM each. For our computations we could ran up to 400 jobs in parallel. The entire computation 
took less than a CPU year in total. 

\begin{table}[htp]
  \begin{center}
    \begin{tabular}{rrrrrrrrrrrrrrrrr}
      \hline
      k/n &     63 &    64 &   65 &    66 \\ 
      \hline
      1   &      0 &     1 &    0 &     0 \\
      2   &      0 &     6 &    0 &     0 \\
      3   &      0 &    41 &    0 &    25 \\ 
      4   &    108 &   557 &   84 &   644 \\
      5   &   4617 & 22267 & 8647 & 46571 \\
      6   &        &       &      &       \\
      7   &        &       &      &       \\
      8   &        &       &      &       \\
      9   &        &       &      &       \\
     10   & 978528 &       &      &       \\
     11   &     28 &704571 &      &       \\
     12   &      1 &     8 &  1 & \\ 
     13   &        &     1 &  0 & 0 \\                      
      \hline 
    \end{tabular}
    \caption{Number of $8$-divisible $\codeshorteff{n}{k}{2}$ codes with minimum distance $24$ -- part 3.}
    \label{tab_8div_p3}
  \end{center}
\end{table}

\begin{theorem}
  \label{thm_classification_2}
  If $C$ is an $8$-divisible $\codeeff{\le 65}{12}{24}{2}$ code, then $C$ is isomorphic to one of the following ten cases:
  \begin{enumerate}
    \item[(1)]
    $\codeeff{n}{k}{d}{q}=\codeeff{63}{12}{24}{2}$\\
    $\begin{pmatrix} 
    001100001110000001111101000011111110010010100100001100000000000\\
    101001111111000000110111010000100110100011011000000010000000000\\
    000100111011100011110111001000010000110000110110100001000000000\\
    010001111111110011001100001001100100010001101000001000100000000\\
    110001110000010111001111011000011100100011000010100000010000000\\
    000000011000110111100011010011101110010001011110000000001000000\\
    010011110001111101010000110100100011101110111111111000000100000\\
    001000110111101100001111110000000001100110000111100000000010000\\
    000111110001100011000000001100011111100001100001111000000001000\\
    000000001111100000111111111100000000011111100000011000000000100\\
    000000000000011111111111111100000000000000011111111000000000010\\
    000000000000000000000000000011111111111111111111111000000000001\\
    \end{pmatrix}$\\
    $W(z)=1z^{0}+630z^{24}+3087z^{32}+378z^{40}$\\
    $\# \operatorname{Aut}=362880$
    \item[(2)]
    $\codeeff{n}{k}{d}{q}=\codeeff{64}{12}{24}{2}$\\
    $\begin{pmatrix}
1000010101010101100101010101011010010101100101011010100101010111\\
0100010101010101011001100110100110100110100110010101011010011011\\
0010000000000000111100110011000011110011000000111100000000111101\\
0001000000000000111100110011110011001100000011000000111100110010\\
0000110000000000001100000000110011110011001100111111001111000000\\
0000001100000000001100000000111100110011110000110011110000111100\\
0000000011000000110000000000111111110000111111000011000011001100\\
0000000000110000000000000011111111001111001100000011111100001100\\
0000000000001100110000000000110011001100110000111111001100111100\\
0000000000000011000000110000111100001111110011000000001111111100\\
0000000000000000000011110000110011111100001100110000110011111100\\
0000000000000000000000001111000000110011111111001111001100111100\\
\end{pmatrix}$\\
    $W(z)=1z^{0}+496z^{24}+3102z^{32}+496z^{40}+1z^{64}$
\item[(3)]
$\codeeff{n}{k}{d}{q}=\codeeff{64}{12}{24}{2}$\\
    $\begin{pmatrix}
1000101010011011110100110000000000000000010101100000110101001101\\
0001000000011111101000100000110110001111000000001100000110010011\\
0000000011110111110110110000110110001111001101110000110011100001\\
0000111110010100110000100000110110001100001000100010010101010100\\
0001010011010000001111000000110110001101001100000001100101001010\\
0001100001101110000000100000000000000011010111101000001111011110\\
0101101001010101101010100000100101011101101100110000000000000000\\
0000000000111111000000001000110111110011011111111000000000000000\\
0000001111111100110000110100101010111111000000011000000000000000\\
0000110011110000110011000010111100101011111001100000000000000000\\
0000110000111111111100110001111011000100100001100000000000000000\\
0011111111111111111111110000000000000011011111111000000000000000\\
\end{pmatrix}$\\
    $W(z)=1z^{0}+496z^{24}+3102z^{32}+496z^{40}+1z^{64}$
\item[(4)]
$\codeeff{n}{k}{d}{q}=\codeeff{64}{12}{24}{2}$\\
    $\begin{pmatrix}
1000010101010101010101101010100101011010011001101001100110011011\\
0100010101011001010110010110010110100110101001010101011001101011\\
0010000000001100000011001111110011001111110000110000000000001101\\
0001000000001100000011111100000000000000111100110011001111001110\\
0000110000000000000000111111000000001100110011111100000011111100\\
0000001100001100000000000000110011001100111111110011111100000000\\
0000000011001100000000110011110000111111110011001100001100000000\\
0000000000110000000011110011110000000000000011111111111111000000\\
0000000000000011000011000000000011110011110011001111111100110000\\
0000000000000000110000001100110000001111001111001100111111110000\\
0000000000000000001100000011110000110000111111000011001111111100\\
0000000000000000000000000000001111111111111100111100001111001100\\
\end{pmatrix}$\\
    $W(z)=1z^{0}+496z^{24}+3102z^{32}+496z^{40}+1z^{64}$
\item[(5)]
$\codeeff{n}{k}{d}{q}=\codeeff{64}{12}{24}{2}$\\
    $\begin{pmatrix}
1001001100101000010111111010100111001011000000000111101111111000\\
0110001011001000000011110011110011000100000000000111110100010000\\
0000101100110001100110000000000000110110000010101111011000011010\\
0000001010101000110101111111000000010100000010101000110000010011\\
0000000001111000000110110000111101110010000000000111101010000110\\
0000011111111000101000001111111100000101000000000000010100110000\\
0000000000000101101010100101001101100110000010011101101010101000\\
0000000000000000111111011100011000110011100011100101111000000000\\
0000000000000000110011110000111111000000010011001011011111100000\\
0000000000000000111100001111111100000000001011111010100110011000\\
0000000000000000110011111111000011001111000110110110000110000000\\
0000000000000011111111001111111111111111000000000000011111111000\\
\end{pmatrix}$\\
    $W(z)=1z^{0}+528z^{24}+3038z^{32}+528z^{40}+1z^{64}$\\ 
    non-projective
\item[(6)]
$\codeeff{n}{k}{d}{q}=\codeeff{64}{12}{24}{2}$\\
    $\begin{pmatrix}
1000011011010011000001010110011000101001100000100001000001100110\\
0001100111110110000000000100000001100000001100011101001100111001\\
0000101011110100001000010101001000010001111001000011000000110011\\
0000111110100001000010100011001001110001100001000000100101010110\\
0001001111111011011010110011001000010101110111000000011001011001\\
0001001011011011000010010100000000000110011000010000000011111111\\
0101010101011010010001010010001110100101110101101000000000000000\\
0000110011001100100011110111010101101001111001000000000000000000\\
0000110000111100011100000111010010010111111000111000000000000000\\
0000111111111111011011001001001100011000100000100000000000000000\\
0000001111001100011000110000111010011000111111011000000000000000\\
0011111111111111011011110110000001100110011110011000000000000000\\
\end{pmatrix}$\\
    $W(z)=1z^{0}+502z^{24}+3087z^{32}+506z^{40}$
\item[(7)]
$\codeeff{n}{k}{d}{q}=\codeeff{64}{12}{24}{2}$\\
    $\begin{pmatrix}
1000110000011001010001010110000011110010000001001011010001100001\\
0000110000011110101001010001111100110010111110110100101101100001\\
0000101010111100010011000110011011100010000110011011100110110111\\
0000000010001011011110011011101101110000001110101101101101101100\\
0000000010001011011001100011101101110000010001001110010010010000\\
0000000010100100111010001110011011101000001000110101001000100100\\
0000110000000110110010001001111011000101010001001101001010010000\\
0100111010001100010001001100111000110000001001011001010010000100\\
0010100010001100010010001101000011110110000110011001001001001000\\
0001111000000001111011011000000000000110000011011001111011000000\\
0000000110011111100011000110011011000000000011011000011011000000\\
0000000001111000000000011110011011110000001111000001100011001100\\
\end{pmatrix}$\\
    $W(z)=1z^{0}+496z^{24}+3102z^{32}+496z^{40}+1z^{64}$
\item[(8)]
$\codeeff{n}{k}{d}{q}=\codeeff{64}{12}{24}{2}$\\
    $\begin{pmatrix}
1000000100111001010000110001011010010100100100001111010110000000\\
0000111100000101011000110001011010100110111110010101111101001010\\
0000000000000000110001010101011100100100010110101101010011011010\\
0100110100110011000000000000000000001110011001101100101011010101\\
0000001101010011000101010100000101100100011000110100000110010011\\
0000000000000000000001101101100101110010011010011011001110010110\\
0000001101010011000001100000110011110001010111111011100100111111\\
0010010000111010000001100000110011011000100100110001010101001001\\
0001100000110101000001100000110011100100011000110000101101000110\\
0000000011111111000000000000000000000000111111110000000111101111\\
0000000000000000000011110000001111111100000011110000000111101111\\
0000000000000000000000000011111111000000111111110001111000001111\\
\end{pmatrix}$\\
    $W(z)=1z^{0}+496z^{24}+3102z^{32}+496z^{40}+1z^{64}$
\item[(9)]
$\codeeff{n}{k}{d}{q}=\codeeff{64}{12}{24}{2}$\\
    $\begin{pmatrix}
1000000010111010000111000000101100000000111000101001010101000111\\
0010101000001011111110100000101111001101110000110110010110111000\\
0000000000110010101100011000011001100110110000000000011001011111\\
0000100010101001011010000100100010001001010100100011010001110101\\
0010000010011001010001001010100101100010000001111000110110000110\\
0000001010110011011110110001100010111111010010111011000010101100\\
0110000000011000000110011000000000000110011001111000001111110011\\
0001100000011000011000011000000111100111100110011000001100000011\\
0000011000011001100000011000000001100000011110011000111100111100\\
0000000110011001111111100000000110000000011001111000000000111100\\
0000000001111000000001111000000111111110000000000000111111110000\\
0000000000000111111111111000000111111111111000000000000000000000\\
\end{pmatrix}$\\
    $W(z)=1z^{0}+496z^{24}+3102z^{32}+496z^{40}+1z^{64}$
    \item[(10)]
    $\codeeff{n}{k}{d}{q}=\codeeff{65}{12}{24}{2}$\\
    $\begin{pmatrix}  
    10000100000000110110010001110100111101010001011110010100000000000\\
    10100100011000001001000110100110111111001000001100011010000000000\\
    01000010011100011000000100110100110000011111011110001001000000000\\
    11110100001110110100000011010110100001011100000100001000100000000\\
    01101011000001100011010001000011001010001111000010111000010000000\\
    00101001110111101011000001011000000110111001001000100000001000000\\
    00011000111111100000111110001000100010001010101001100000000100000\\
    00000111001011100111110001010100000001111001100000011000000010000\\
    00011111000111100000001111001101111110000111100111111000000001000\\
    00000000111111100000000000111100011110000000011111111000000000100\\
    00000000000000011111111111111100000001111111111111111000000000010\\
    00000000000000000000000000000011111111111111111111111000000000001\\
    \end{pmatrix}$\\
    $W(z)=1z^{0}+390z^{24}+3055z^{32}+650z^{40}$\\
    $\# \operatorname{Aut}=15600$
  \end{enumerate}
  There is a unique $8$-divisible $\codeeff{\le 66}{13}{24}{2}$ code, see the $\codeeff{64}{13}{24}{2}$ code at the beginning of Section~\ref{sec_computer_classification}.   
  No $8$-divisible $\codeeff{\le 67}{\ge 14}{24}{2}$ code exists.
\end{theorem}

For some parameters $n$ and $k$ there exists a unique code that eventually admits an easy description. We give a few examples. For dimensions $1\le k\le 3$ the $8$-divisible optimal codes 
are more or less trivial. The $\codeeff{45}{4}{24}{2}$ is given by the points of a solid. The $\codeeff{51}{8}{24}{2}$ code is obtained via the concatenation of an ovoid in 
$\operatorname{PG}(3,\mathbb{F}_4)$ with the binary $[3,2]$ simplex code \cite[Lemma 24]{honold2018partial}. Note that this code is a two-weight code with weights $24$ and $32$.

In some cases the $8$-divisible codes attain the maximal possible minimum distance $d=24$ for $\codeshorteff{n}{k}{2}$ codes. In Table~\ref{tab_opt} we list for dimensions 
$k\le 13$ the lengths $n$ and the corresponding counts for which the maximum, using the bounds from \url{www.codetables.de} \cite{Grassl:codetables}, is attained. 
We remark that, according to those tables, for $\codeshorteff{61}{11}{2}$ codes it is unknown whether minimum distance $25$ can be achieved. Similarly, for $\codeshorteff{63}{12}{2}$ it is unknown whether 
the minimum distance $25$ or $26$ can be attained. In Section~\ref{sec_class_optimal} in the appendix we completely list the generator matrices and key parameters of the corresponding codes. 
We remark that if a linear code over $\mathbb{F}_q$ meets the Griesmer bound and the minimum distance is divisible by $q^r$, where $r\in\mathbb{Q}$, then the weight of each 
codeword is divisible by $q^r$, see \cite[Theorem 1]{ward1998divisibility}.  

\begin{proposition}
  \begin{enumerate}
    \item[(1)] Every $\codeeff{\le 62}{k}{\{24,32\}}{2}$ code satisfies $k\le 8$. The counts for dimension $k=8$ are given by $\codeshorteff{51}{8}{2}$: 1, $\codeshorteff{54}{8}{2}$: 1, 
               $\codeshorteff{55}{8}{2}$: 2, $\codeshorteff{56}{8}{2}$: 3,$\codeshorteff{57}{8}{2}$: 11, $\codeshorteff{58}{8}{2}$: 13, $\codeshorteff{59}{8}{2}$: 33,  
               and $\codeshorteff{60}{8}{2}$: 12. 
    \item[(2)] Every $\codeeff{\le 63}{k}{\{24,32,56\}}{2}$ code satisfies $k\le 9$. For dimension $k=9$ there exist only two non-isomorphic $\codeeff{56}{9}{\{24,32,56\}}{2}$ codes, which 
               both contain a unique codeword of weight $56$.
  \end{enumerate}  
\end{proposition} 

In \cite[Lemma 2.2]{pignatelli2007wahl} it has been proven that each $\codeeff{\le 67}{k}{\{24,32,56\}}{2}$ code has dimension $k\le 10$, see 
also \cite[Lemma 2.1]{pignatelli2007wahl} and \cite[Lemma 2.6]{wahl1998nodes} for the two-weight code case $W=\{24,32\}$. 

\begin{table}[htp]
  \begin{center}
    \begin{tabular}{rl}
      \hline
      k & n \\
      \hline
      1 & 24:1\\ 
      2 & 36:1\\
      3 & 42:1, 44:1\\ 
      4 & 45:1, 46:1, 47:2, 48:4\\
      5 & 47:1, 48:4, 49:1, 50:6\\
      6 & 48:1, 49:1, 50:2, 51:5\\
      7 & 50:1, 51:1, 52:4, 53:7, 54:58\\
      8 & 51:1, 54:1, 55:4, 56:55\\
      9 & 56:2 \\ 
      \hline 
    \end{tabular}
    \caption{Number of optimal $8$-divisible codes per dimension and length.}
    \label{tab_opt}
  \end{center}
\end{table}

While the possible lengths of $q^r$-divisible linear codes over $\mathbb{F}_q$ have been completely characterized in \cite[Theorem 4]{kiermaier2017improvement}, see also 
Section~\ref{sec_theoretical}, the problem becomes harder if one restricts to projective codes or prescribes the dimension. A few partial results in that direction have 
been obtained in \cite{ubt_eref40887,honold2018partial}. An upper bound on the maximum possible dimension of a $\Delta$-divisible linear code was proven in \cite{ward1992bound}.

\section{Codes of nodal surfaces}
\label{sec_nodal_surfaces}

The codes of nodal surfaces with degree $s$ and the maximum number $m=\mu(s)$ of nodes are more or less trivial for $s\le 5$. For $s=3$ the code is a $\codeeff{4}{1}{4}{2}$ code and spanned 
by a single codeword of weight $4$. For $s=4$ the code is a $\codeeff{16}{5}{8}{2}$ code with weight enumerator $W(z)=1z^{0}+30z^{8}+1z^{16}$, which corresponds to the points of an 
affine solid. For $s=5$ the code is a $\codeeff{31}{5}{16}{2}$ code with weight enumerator $W(z)=1z^{0}+31z^{16}$, which corresponds to the points of $\mathbb{F}_2^5$, i.e., the simplex 
code $\mathcal{S}(5)$. The situation changes for $s=6$. From a general upper bound $m=\mu(6)\le 66$ can be concluded. The dimension argument mentioned in the introduction gives 
$k\ge m-53$, i.e., $k\ge 13$ for $m=66$ and $k\ge 12$ for $m=65$. The codes of sextics, i.e., nodal surfaces of degree $s=6$ have a minimum distance $d\ge 24$ and are $8$-divisible. 
In \cite[Section 7]{jaffe1997sextic} it is shown that there is no codeword of weight $48$. A codeword of weight $64$ can only be contained if the dimension of the code is 
$k=11$, see \cite[Section 9]{jaffe1997sextic}. So, for $m\in\{65,66\}$ there cannot be a codeword of weight $64$. In \cite[Theorem 1.6]{catanese2005even} it is shown that there is no 
codeword of weight $64$ in a code corresponding to a sextic normal surface with only rational double points as singularities. Thus, for $m\ge 65$ the weights are contained in 
$\{24,32,40,56\}$. For every weight $w\in \{24,32,40,56\}$ there is a sextic whose corresponding code contains a codeword of weight $w$, see \cite{catanese2005even}. 
Obviously, each $\codeeff{n}{k}{24}{2}$ code with at least two codewords of weight $56$, i.e., $a_{56}\ge 2$, satisfies $n\ge 56+24/2=68$. Thus, in order to classify the 
$\codeeff{n}{\ge 12}{\{24,32,40,56\}}{2}$ codes, it satisfies to classify the $\codeeff{n}{\ge 11}{\{24,32,40\}}{2}$ codes and to eventually enlarge them with a unique codeword  
of weight $56$. Using the algorithmic approach presented in Section~\ref{sec_computer_classification} we obtain the counts stated in Table~\ref{tab_24_32_40_p1} and Table~\ref{tab_24_32_40_p2}.   

\begin{table}[htp]
  \begin{center}
    \begin{tabular}{rrrrrrrrrrrrrrrrr}
      \hline
      k/n & 24 & 32 & 36 & 40 & 42 & 44 & 45 & 46 & 47 & 48 & 49 & 50 & 51 & 52 & 53 &  54 \\ 
      \hline
      1   &  1 &  1 &  0 &  1 &  0 &  0 &  0 &  0 &  0 &  0 &  0 &  0 &  0 &  0 &  0 &   0 \\ 
      2   &    &    &  1 &  1 &  0 &  2 &  0 &  0 &  0 &  2 &  0 &  0 &  0 &  2 &  0 &   0 \\
      3   &    &    &    &    &  1 &  1 &  0 &  2 &  0 &  3 &  0 &  3 &  0 &  5 &  0 &   6 \\ 
      4   &    &    &    &    &    &    &  1 &  1 &  2 &  3 &  1 &  4 &  5 & 13 &  5 &  20 \\
      5   &    &    &    &    &    &    &    &    &  1 &  3 &  1 &  6 &  5 & 28 & 15 &  85 \\
      6   &    &    &    &    &    &    &    &    &    &  1 &  1 &  2 &  5 & 20 & 29 & 153 \\
      7   &    &    &    &    &    &    &    &    &    &    &    &  1 &  1 &  4 &  7 &  54 \\
      8   &    &    &    &    &    &    &    &    &    &    &    &    &  1 &  0 &  0 &   1 \\                    
      \hline 
    \end{tabular}
    \caption{Number of $\codeshorteff{n}{k}{2}$ codes with weights in $\{24,32,40\}$ -- part 1.}
    \label{tab_24_32_40_p1}
  \end{center}
\end{table}

\begin{table}[htp]
  \begin{center}
    \begin{tabular}{rrrrrrrrrrrrrrrrr}
      \hline
      k/n &  55 &   56 &   57 &    58 &    59 &     60 &      61 &      62 &    63 & 64 \\ 
      \hline
      3   &   0 &    7 &      &       &       &        &         &         &       &    \\ 
      4   &  16 &   43 &   13 &       &       &        &         &         &       &    \\
      5   &  80 &  321 &  180 &   784 &       &        &         &         &       &    \\
      6   & 286 & 1557 & 2026 & 10360 & 14011 &        &         &         &       &    \\
      7   & 130 & 1176 & 3604 & 31470 & 91163 & 650496 &         &         &       &    \\
      8   &   3 &   17 &   61 &  1127 & 10631 & 247845 & 1818544 &         &       &    \\
      9   &     &      &      &     3 &    14 &    400 &   18024 & 1270327 &       &    \\
     10   &     &      &      &       &       &      3 &       7 &     394 & 77954 &    \\
     11   &     &      &      &       &       &        &         &       1 &     9 & 47 \\                      
      \hline 
    \end{tabular}
    \caption{Number of $\codeshorteff{n}{k}{2}$ codes with weights in $\{24,32,40\}$ -- part 2.}
    \label{tab_24_32_40_p2}
  \end{center}
\end{table}
 
We remark that no $11$-dimension binary linear code with weights in $\{24,32,40\}$ can be extended with a codeword of weight $56$. 
Computing the $12$- and $13$-dimensional binary linear code with weights in $\{24,32,40\}$ we can state:  

\begin{theorem}
  \label{thm_classification_1}
  If $C$ is a $\codeeff{\le 65}{12}{\{24,32,40,56\}}{2}$ code, then $C$ is isomorphic to one of the following three cases:
  \begin{enumerate}
    \item[(1)]
    $\codeeff{n}{k}{d}{q}=\codeeff{63}{12}{24}{2}$\\
    $\begin{pmatrix} 
    001100001110000001111101000011111110010010100100001100000000000\\
    101001111111000000110111010000100110100011011000000010000000000\\
    000100111011100011110111001000010000110000110110100001000000000\\
    010001111111110011001100001001100100010001101000001000100000000\\
    110001110000010111001111011000011100100011000010100000010000000\\
    000000011000110111100011010011101110010001011110000000001000000\\
    010011110001111101010000110100100011101110111111111000000100000\\
    001000110111101100001111110000000001100110000111100000000010000\\
    000111110001100011000000001100011111100001100001111000000001000\\
    000000001111100000111111111100000000011111100000011000000000100\\
    000000000000011111111111111100000000000000011111111000000000010\\
    000000000000000000000000000011111111111111111111111000000000001\\
    \end{pmatrix}$\\
    $W(z)=1z^{0}+630z^{24}+3087z^{32}+378z^{40}$\\
    $\# \operatorname{Aut}=362880$
    \item[(2)]
    $\codeeff{n}{k}{d}{q}=\codeeff{64}{12}{24}{2}$\\
    $\begin{pmatrix}
    0000110001101110000100100100100011011000011011011110100000000000\\
    1011110000100110010000001100010000111101001110111000010000000000\\
    1010110001001010110010000000101111110000001100101011001000000000\\
    1111100000001100000010100100111101000011011011101000000100000000\\
    0111000000001010110110001100011000000110111100110011000010000000\\
    0000000100001001111110011010010101001101010101010101000001000000\\
    0101011111010000010001111001110011000100100000101100000000100000\\
    0011010011001000001111111001111111011111010011011100000000010000\\
    0000101111000110000000000111101111000011001111000011000000001000\\
    0000011111000001111111111111100000111111000000111111000000000100\\
    0000000000111111111111111111100000000000111111111111000000000010\\
    0000000000000000000000000000011111111111111111111111000000000001\\
    \end{pmatrix}$\\
    $W(z)=1z^{0}+502z^{24}+3087z^{32}+506z^{40}$\\
    $\# \operatorname{Aut}=5760$
    \item[(3)]
    $\codeeff{n}{k}{d}{q}=\codeeff{65}{12}{24}{2}$\\
    $\begin{pmatrix}  
    10000100000000110110010001110100111101010001011110010100000000000\\
    10100100011000001001000110100110111111001000001100011010000000000\\
    01000010011100011000000100110100110000011111011110001001000000000\\
    11110100001110110100000011010110100001011100000100001000100000000\\
    01101011000001100011010001000011001010001111000010111000010000000\\
    00101001110111101011000001011000000110111001001000100000001000000\\
    00011000111111100000111110001000100010001010101001100000000100000\\
    00000111001011100111110001010100000001111001100000011000000010000\\
    00011111000111100000001111001101111110000111100111111000000001000\\
    00000000111111100000000000111100011110000000011111111000000000100\\
    00000000000000011111111111111100000001111111111111111000000000010\\
    00000000000000000000000000000011111111111111111111111000000000001\\
    \end{pmatrix}$\\
    $W(z)=1z^{0}+390z^{24}+3055z^{32}+650z^{40}$\\
    $\# \operatorname{Aut}=15600$
  \end{enumerate}
  No $\codeeff{\le 66}{\ge 13}{\{24,32,40,56\}}{2}$ code exists.
\end{theorem}  
Of course Theorem~\ref{thm_classification_1} is implied by Theorem~\ref{thm_classification_2}. Thus, we can also allow codewords of 
weight $48$ without changing the result of Theorem~\ref{thm_classification_1}.

So, we have computationally reproven $\mu(6)<66$, c.f.\ \cite{jaffe1997sextic}. More precisely, \cite[Theorem 8.1]{jaffe1997sextic} and 
\cite[Theorem A]{pignatelli2007wahl} show that no $\codeeff{\le 66}{13}{\{24,32,40,56\}}{2}$ code exists. For $m=65$ nodes we have extracted an exhaustive list of three possible 
candidates of codes. 
Having our 
classification at hand it is pretty easy to determine the corresponding code since number (3) is the unique code that admits an automorphism of 
order $5$ without a fixed point - a property that also applies to the Barth sextic. It would be nice to have a short tailored argument 
to show that codes number (1) and (2) cannot correspond to a nodal surface. A computer verification of that fact is presented in appendix~\ref{sec_extended_code}. 
As a consequence, the code of each nodal sextic with $65$ nodes is given by (3). Indeed, the Barth sextic is a  member of a $3$-parameter family 
of nodal sextics with $65$ nodes, see \cite[Theorem 5.5.9]{PhdPettersen}.

Up to isomorphism there exists a unique $\codeshorteff{64}{11}{2}$ subcode that can be obtained via shortening:

\noindent
$\codeeff{n}{k}{d}{q}=\codeeff{64}{11}{24}{2}$\\
$\begin{pmatrix}  
    1000001010000011100001101101100100000011011111100011010000000000\\
    1101001010001111101000010100010010011000001101000011101000000000\\
    0010001001001110010000110001000011010110011111100010100100000000\\
    1111011010000100010101101101000011001001000011000010100010000000\\
    0111100101100001001011000000110111000111000000100101100001000000\\
    0001000100111000111111010011110010000011000110001100000000100000\\
    0000011100011100011111000000001110111100100100011100000000010000\\
    0000000011100100000111001111111110000011100010000011100000001000\\
    0000111111100011111111000000000001111111100001111111100000000100\\
    0000000000011111111111000000000000000000011111111111100000000010\\
    0000000000000000000000111111111111111111111111111111100000000001\\ 
\end{pmatrix}$\\
$W(z)=1z^{0}+246z^{24}+1551z^{32}+250z^{40}$\\
$\# \operatorname{Aut}=240$\\

\section{Theoretical arguments}
\label{sec_theoretical}

The classification results from Section~\ref{sec_computer_classification} and Section~\ref{sec_nodal_surfaces} have been obtained 
by extensive computer calculations, so that it would be nice to have short theoretical arguments for some of these findings.
First we note that the MacWilliams identities of a $\codeshorteff{n}{k}{2}$ code, see Equation~(\ref{eq_mac_williams}), for the coefficients of 
$y^0$, $y^1$, $y^2$, and $y^3$ can be rewritten to (see also \cite{wahl1998nodes}):
\begin{eqnarray}
  \sum_{i>0} a_i &=& 2^k-1,\\
  \sum_{i\ge 0} ia_i &=& 2^{k-1}n,\\
  \sum_{i\ge 0} i^2a_i &=& 2^{k-1}(a_2^*+n(n+1)/2),\\ 
  \sum_{i\ge 0} i^3a_i &=& 2^{k-2}(3(a_2^*n-a_3^*)+n^2(n+3)/2).
\end{eqnarray} 
We also speak of the first four MacWilliams identities. In this special form, those equations are also known as the 
first four (Pless) power moments \cite{pless1963power}.

\begin{lemma}
  \label{lemma_n_ge_63}
  Let $C$ be a binary $8$-divisible linear code with minimum distance $d\ge 24$, dimension $k=12$ and effective length $n\le 65$, then
  $a_{40}\ge 1$ and $n\ge 63$.
\end{lemma}
\begin{proof}
  Solving the first four MacWilliams identities for $a_{24}$, $a_{32}$, $a_{40}$, and $a_{48}$ gives
  $$
    a_{40}=\frac{205}{2}n^2-6808n-\frac{1}{2}n^3+(208-3n)a_2^*+3a_3^*+6a_{56}+20a_{64}+147420 
  $$
  and
  $$
    a_{40}+a_{48}=71n^2-\frac{14504}{3}n-\frac{1}{3}n^3+(144-2n)a_2^*+2a_3^*+2a_{56}+10a_{64}+106470.
  $$    
  Since $a_2^*,a_3^*,a_{56},a_{64}\ge 0$, $208-3n\ge 0$, $144-2n\ge 0$ we have 
  $$
    a_{40}\ge \frac{205}{2}n^2-6808n-\frac{1}{2}n^3+147420  
  $$
  and   
  $$
    a_{40}+a_{48}\ge 71n^2-\frac{14504}{3}n-\frac{1}{3}n^3+106470.  
  $$  
  For $54\le n\le 60$ we have $a_{40}+a_{48}<0$, which is impossible. If either $n\le 53$ or $61\le n\le 65$, 
  then $a_{40}\ge 1$. Thus, $a_{40}\ge 1$. Consider the residual code $C'$ of a codeword of weight $40$. $C'$ has dimension $11$ 
  and is doubly-even, i.e., its length is at least $23$.  
\end{proof}
We remark that all lengths $63\le n\le 65$ can be attained by suitable codes, see Theorem~\ref{thm_classification_1}. Next we look at the 
restrictions that are implied solely by $q^r$-divisibility of a code.

\begin{lemma}
  \label{lemma_divisibility}
  (\cite[Lemma 7]{honold2018partial})\\ 
  Let $C$ be a $q^r$ divisible $\codeshorteff{n}{k}{q}$ code and $\mathcal{P}$ be the corresponding multiset of points in $\mathbb{F}_q^k$. Then for $0\le l\le \min(k-1,r)$ 
  let $\mathcal{P}'$ be the set of points that is contained in an arbitrary $(k-l)$-dimensional subspace of $\mathbb{F}_q^k$ and $C'$ be the corresponding 
  linear code. With this, the code $C'$ is $q^{r-l}$-divisible.
\end{lemma}  
As a consequence the effective length of $C'$ is divisible by $q^{k-l}$, which is perfectly reflected by the first three rows of tables 
(\ref{tab_8div_p1})-(\ref{tab_8div_p3}) and (\ref{tab_24_32_40_p1})-(\ref{tab_24_32_40_p2}).

\begin{lemma}
	\label{lemma:snumb_card}
	(\cite[Lemma 6]{kiermaier2017improvement})\\
	For $r\in\mathbb{N}_0$ and $i\in\{0,\ldots,r\}$, there is a $q^r$-divisible $\codeshorteff{n}{k}{q}$ code with suitable dimension $k$ and effective length 
	\[
		n=\snumb{r}{i}{q}
		:=\frac{q^{r+1}-q^i}{q-1}
		=\sum_{j=i}^r q^{j}
		=q^i + q^{i+1} + \ldots + q^r\text{.}
	\]
\end{lemma}

The numbers $\snumb{r}{i}{q}$ have the property that they are divisible by $q^i$, but not by $q^{i+1}$.
This allows us to create kind of a positional system upon the sequence of base numbers
\[
	S_q(r) = (\snumb{r}{0}{q}, \snumb{r}{1}{q},\ldots, \snumb{r}{r}{q})\text{.}
\]

\begin{lemma}
\label{lem:unique_representation}
(\cite[Lemma 7]{kiermaier2017improvement})\\
Let $n\in\mathbb{Z}$ and $r\in\mathbb{N}_0$.
There exist $a_0,\dots,a_{r-1}\in\{0,1,\dots,q-1\}$ and $a_r\in\mathbb{Z}$ with $n=\sum_{i=0}^r a_i \snumb{r}{i}{q}$.
Moreover this representation is unique.
\end{lemma}
 
The unique representation $n = \sum_{i=0}^r a_i \snumb{r}{i}{q}$ of Lemma~\ref{lem:unique_representation} will be called the \emph{$S_q(r)$-adic expansion} of $n$. 
The number $a_r$ will be called the \emph{leading coefficient} 
of the $S_q(r)$-adic expansion.  

\begin{theorem}
  \label{thm:characterization}
  (\cite[Theorem 4]{kiermaier2017improvement})\\
  Let $n\in\mathbb{Z}$ and $r\in\mathbb{N}_0$.
  The following are equivalent:
  \begin{enumerate}[(i)]
  \item\label{thm:characterization:card} There exists a $q^r$-divisible $\codeshorteff{n}{k}{q}$ for a suitable dimension $k$.
  \item\label{thm:characterization:n_strong} The leading coefficient of the $S_q(r)$-adic expansion of $n$ is non-negative.
  \end{enumerate}
\end{theorem} 

\begin{lemma}
  \label{lemma_no_4_div}
  There is no binary $4$-divisible linear code with an effective length $n\in\{1,2,3,5,9\}$. 
\end{lemma}
\begin{proof}
  We have $\snumb{2}{0}{2}=7$, $\snumb{2}{1}{2}=6$, and $\snumb{2}{2}{2}=4$, so that we have the following 
  $S_2(2)$-adic expansions of $n\in\{1,2,3,5,9\}$:
  \begin{itemize}
    \item $1=-3\cdot 4+1\cdot 6+1\cdot 7$,
    \item $2=-2\cdot 4+1\cdot 6+0\cdot 7$,
    \item $3=-1\cdot 4+0\cdot 6+1\cdot 7$,
    \item $5=-2\cdot 4+1\cdot 6+1\cdot 7$,
    \item $9=-1\cdot 4+1\cdot 6+1\cdot 7$.
  \end{itemize}
  Note that the leading coefficient is negative in all cases and apply Theorem~\ref{thm:characterization}.  
\end{proof}

Restrictions on the dimension can be incorporated via \emph{residual codes}.
\begin{lemma}
  \label{lemma_residual}
  Let $C$ be an $\codeshorteff{n}{k}{q}$ code and $u\in C$ be a codeword of weight $w$. Let $C_1$ be the code generated by the codewords of $C$ 
  restricted to those coordinates that are not contained in the support $\operatorname{supp}(w)$ and $C_2$ be the code generated by the codewords of $C$ 
  restricted to those coordinates that are contained in $\operatorname{supp}(w)$. Then, we have $\dim(C_1)+\dim(C_2)=k$ and the effective lengths 
  are given by $n-w$ and $w$.  
\end{lemma}
The code $C_1$ is called the \emph{residual} code of $C$ with respect to $u$. Note that if $w$ is smaller than twice the minimum distance of $C$, then 
$\dim(C_2)=1$ and $\dim(C_1)=k-1$. If $w=2d$, e.g., $w=48$ in our application, then a complete classification of the $\codeeff{w}{k'}{\{d,2d\}}{q}$ codes is known, see 
\cite{jungnickel2018classification}. If $C$ is $q^r$-divisible, then $C_1$ and $C_2$ are $q^{r-1}$-divisible. The decomposition of $C$ into codes $C_1$ and 
$C_2$ is the inverse of the so-called \emph{construction X}, see e.g.~\cite[Ch. 18, Theorem 9]{macwilliams1977theory}.

\begin{proposition}
  Let $C$ be a binary $8$-divisible linear code with minimum distance $d\ge 24$, dimension $k=12$ and effective length $n\le 65$, then:
  \begin{enumerate}
    \item[(1)] If $C$ contains a word $c_{64}$ of weight $64$, then $n=64$ and the other codewords have weights in $\{24,32,40\}$.
    \item[(2)] If $C$ contains a word $c_{56}$ of weight $56$, then $a_{56}=1$, $a_{64}=0$, and $n\in\{63,64\}$. 
  \end{enumerate}
\end{proposition}
\begin{proof}
  Due to Lemma~\ref{lemma_n_ge_63} we can assume $n\ge 63$.
  \begin{enumerate}
    \item[(1)] Clearly $n\ge 64$. By considering the residual code of $c_{64}$, Lemma~\ref{lemma_no_4_div} shows that $n=65$ is impossible.
               In $\mathbb{F}_2^{64}$ the sum of $c_{64}$ and a codeword of weight $48$ or $56$ is $16$ or $8$, respectively. Clearly the codeword 
               of weight $64$ is unique.
    \item[(2)] By considering the residual code of $c_{56}$, Lemma~\ref{lemma_no_4_div} shows that $n=65$ is impossible. As shown in (1), there is no 
               codeword of weight $64$. Due to $d\ge 24$ two codewords of weight $56$ have to intersect in at least $44$ positions, which would imply $n\ge 68$. 
               Thus, there is a unique codeword of weight $56$. If there is a codeword $c_{48}$ of weight $48$, then $n=64$ and the supports of $c_{56}$ and 
               $c_{48}$ intersect in a set of cardinality $40$.
  \end{enumerate}                
\end{proof}

\section*{Acknowledgments}
The author likes to thank Iliya Bouyukliev for discussions on the usage of his software \texttt{Q-Extension} and a personalized version that is capable to deal 
with larger dimensions. He also benefited from many discussion with Alfred Wassermann and Michael Kiermaier. Further thanks go to Fabrizio Catanese for bringing 
the problem of nodal surfaces with many nodes to our attention and to point out that symmetries of a nodal surface carry over to the associated binary code. 
The author would also like to thank the \textit{High Performance Computing group} of 
the University of Bayreuth for providing the excellent computing cluster and especially Bernhard Winkler for his support.  


\appendix

\section{The extended code of a nodal sextic}
\label{sec_extended_code}
Actually, there are two codes associated with a nodal surface. Some authors, see e.g.\ \cite{endrass1997minimal}, speak of even sets of nodes in the 
geometric context, which can be distinguished into strictly even nodes and weakly even nodes. The corresponding codes are called the (associated) code 
$\mathcal{K}$ of the nodal surface and the extended code $\mathcal{K}'$. For nodal sextics with 65 ordinary double points $\mathcal{K}$ can 
only be one of the three possibilities in Theorem~\ref{thm_classification_1}. The extended code $\mathcal{K}'$ contains $\mathcal{K}$ as a subcode and 
the lower bound for the dimension of $\mathcal{K}'$ is one larger than for $\mathcal{K}$. For sextics one we additionally 
know that the weights of $\mathcal{K}'$ are $4$-divisible and have minimum distance at least $16$, see e.g.\ \cite[Theorem 1.10]{endrass1997minimal}. 
Moreover, $\mathcal{K}'\backslash \mathcal{K}$ does not contain codewords of weight $20$ or $24$, see \cite[Corollary 1.11]{endrass1997minimal}. 
This motivates the following coding theoretic statement:
\begin{proposition}
  \label{prop_max_weight_extended_44}
  Let $\mathcal{K}$ be one of the codes of Theorem~\ref{thm_classification_1} and $\mathcal{K}'$ be a $(\dim(\mathcal{K})+1)$-dimensional binary code containing $\mathcal{K}$ 
  as a subcode such that the weights of the  codewords in $\mathcal{K}'\backslash\mathcal{K}$ are $4$-divisible, at least $16$ and not equal to $20$ or $24$. 
  If the effective length $\neff(\mathcal{K}')$ of $\mathcal{K}'$ satisfies $\neff(\mathcal{K})<\neff(\mathcal{K}')\le 66$, then $\mathcal{K}$ is the code of effective 
  length $65$ in Theorem~\ref{thm_classification_1} and the maximum weight in $\mathcal{K}'$ is exactly $44$.
\end{proposition}
\begin{proof}
  We proof the statement computationally using integer linear programming. To that end let $n$ be the effective length of $\mathcal{K}$ and $c'$ be 
  a codeword with $\langle \mathcal{K},c'\rangle=\mathcal{K}'$, such that $\neff(\mathcal{K}')=n+\delta$, where $1\le \delta\le 66-n$. By assumption the entries 
  of $c'$ at position $n+i$ are equal to $1$ for $1\le i\le\delta$. We model $c'$ 
  by the binary variables $x_i$ for $1\le i\le n$, i.e., the $i$th component of $c'$ equals $x_i$. If $c'$ has weight $\gamma$, $c\in\mathcal{K}$ has 
  weight $\beta$, and the number of common ones of $c$ and $c'$ is $\alpha$, then $c'+c\in\mathcal{K}'\backslash \mathcal{K}$ has weight 
  $\gamma+\beta-2\alpha$. If $\Lambda$ is an upper bound for the weight of a codeword in $\mathcal{K}'\backslash{K}$, then
  $$
    \frac{\gamma+\beta}{2}-\frac{\Lambda}{2} \le \alpha\le \frac{\gamma+\beta}{2}-8
  $$ 
  due to the minimum distance of $\mathcal{K}'$, where $\alpha=\sum_{1\le i\le n\,:\, c_i=1} x_i$ and $\beta=\operatorname{wt}(c)$. 
  In order to model the gap in the weight spectrum, i.e., if $c'+c$ does not has weight $16$ then the weight is at least $28$, we introduce 
  the binary variable $y_c$ and require
  \begin{equation}
    \label{ie_codeword}
    \frac{\gamma+\operatorname{wt}(c)}{2}-8-\left(\tfrac{\Lambda}{2}-8\right)\cdot y_c \,\le\,\sum_{1\le i\le n\,:\, c_i=1} x_i \,\le\, \frac{\gamma+\operatorname{wt}(c)}{2}-8-6y_c,\\
  \end{equation}  
  for all $c\in\mathcal{K}$ with $\operatorname{wt}(c)\neq 0$. If $y_c=0$ then these conditions are equivalent to $\operatorname{wt}(c'+c)=16$ and 
  to $28\le \operatorname{wt}(c'+c)\le\Lambda$ otherwise. Additionally we use the constraint 
  $\sum_{i=1}^{n}x_i=\gamma-\delta$, the target function $\sum_{i=1}^{n}ix_i$, and denote the corresponding integer linear program by 
  $\operatorname{ILP}_{\gamma,\Lambda,\delta,\mathcal{K}}$. 

  If for a given $\mathcal{K}$ a code $\mathcal{K}'$, satisfying the mentioned restrictions, exists, then $\operatorname{ILP}_{\gamma,\gamma,\delta,\mathcal{K}}$ 
  has a solution, where $\gamma$ is the maximum weight in $\mathcal{K}'\backslash\mathcal{K}$. Computationally we check that 
  for $\gamma\in\{16,28,32,\dots,64\}$ $\operatorname{ILP}_{\gamma,\gamma,\delta,\mathcal{K}}$ is feasible if and only if $\gamma=44$, $\delta=1$, and $\mathcal{K}$  
  has effective length $65$.    
\end{proof}
We remark that our ILP formulation is only a relaxation of the original problem for $\mathcal{K}'$, e.g., $\operatorname{wt}(c+c')=30\not\equiv 0\pmod 4$ 
is not excluded by inequality~(\ref{ie_codeword}). As a relaxation, we may ignore those constraints for some codewords $c\in\mathcal{K}$ or use the symmetry group of $\mathcal{K}$ (cf.\ 
the proof of Theorem~\ref{thm_unique_extended_code}). Since all ILPs can be solved in a few hours, which is negligible to the running times required in 
Section~\ref{sec_computer_classification}, we do not go into details here. 

As an example we spell out the details of $\operatorname{ILP}_{44,44,1,\mathcal{K}}$, where $\mathcal{K}$ has effective length $65$:
\begin{eqnarray}
\sum_{i=1}^{65}ix_i && \text{subject to}\label{ilp_ex}\\
\sum_{i=1}^{65}x_i=43 &&\nonumber\\
6y_c+\sum_{1\le i\le 65\,:\, c_i=1} x_i \le 34 && \forall c\in\mathcal{K}:\operatorname{wt}(c)=40,\nonumber\\
14y_c+\sum_{1\le i\le 65\,:\, c_i=1} x_i \ge 34 && \forall c\in\mathcal{K}:\operatorname{wt}(c)=40,\nonumber\\
6y_c+\sum_{1\le i\le 65\,:\, c_i=1} x_i \le 30 && \forall c\in\mathcal{K}:\operatorname{wt}(c)=32,\nonumber\\
14y_c+\sum_{1\le i\le 65\,:\, c_i=1} x_i \ge 30 && \forall c\in\mathcal{K}:\operatorname{wt}(c)=32,\nonumber\\
6y_c+\sum_{1\le i\le 65\,:\, c_i=1} x_i \le 26 && \forall c\in\mathcal{K}:\operatorname{wt}(c)=24,\nonumber\\
14y_c+\sum_{1\le i\le 65\,:\, c_i=1} x_i \ge 26 && \forall c\in\mathcal{K}:\operatorname{wt}(c)=24,\nonumber\\
x_i\in\{0,1\}&&\forall 1\le i\le 65,\nonumber\\ 
y_c\in\{0,1\}&&\forall c\in\mathcal{K}:\operatorname{wt}(c)\in\{24,32,40\}.\nonumber
\end{eqnarray}

We remark that in the general geometric context $\mathcal{K}'=\mathcal{K}$ is possible, which is excluded by $\dim(\mathcal{K})<13$ in our situation. Thus, 
Proposition~\ref{prop_max_weight_extended_44} applies in the case of a nodal sextic with $65$ ordinary double points, i.e., $\mathcal{K}$ has 
effective length $65$ and is uniquely characterized in Theorem~\ref{thm_classification_1}. We can even uniquely classify $\mathcal{K}'$:
\begin{theorem}
  \label{thm_unique_extended_code}
  Let $\mathcal{K}$ and $\mathcal{K}'$ be as in Proposition~\ref{prop_max_weight_extended_44}, then $\mathcal{K}'$ is given by
  $$
    \begin{pmatrix}  
    100001000000001101100100011101001111010100010111100101000000000000\\
    101001000110000010010001101001101111110010000011000110100000000000\\
    010000100111000110000001001101001100000111110111100010010000000000\\
    111101000011101101000000110101101000010111000001000010001000000000\\
    011010110000011000110100010000110010100011110000101110000100000000\\
    001010011101111010110000010110000001101110010010001000000010000000\\
    000110001111111000001111100010001000100010101010011000000001000000\\
    000001110010111001111100010101000000011110011000000110000000100000\\
    000111110001111000000011110011011111100001111001111110000000010000\\
    000000001111111000000000001111000111100000000111111110000000001000\\
    000000000000000111111111111111000000011111111111111110000000000100\\
    000000000000000000000000000000111111111111111111111110000000000010\\
    011100000011000001001001000011010000000010000010000010010000000001\\
    \end{pmatrix}.
  $$
\end{theorem}   
\begin{proof}
  First we note that the weight enumerator of $\mathcal{K}'\backslash\mathcal{K}$ is given by
  $W(z)=26z^{16}+650z^{28}+1690z^{32}+1300z^{36}+300z^{40}+130z^{44}$ and $\mathcal{K'}$ is a $\codeeff{13}{66}{16}{2}$ code, i.e., 
  all conditions for $\mathcal{K}'$ are satisfied.

  From Proposition~\ref{prop_max_weight_extended_44} we conclude that $\mathcal{K}$ is the code of effective length $65$ in Theorem~\ref{thm_classification_1} 
  and that $\mathcal{K}'$ has maximum weight $44$, which is indeed attained. Now we add the constraints $y_c=1$ to the ILP formulation~(\ref{ilp_ex}) 
  for all $c\in\mathcal{K}:\operatorname{wt}(c)\in\{24,40\}$, i.e., we require $\operatorname{wt}(c+c')\neq 16$. Since this ILP does not have a solution, 
  we can conclude that $\mathcal{K}'\backslash\mathcal{K}$ contains a codeword of weight $16$.
  
  Next we consider the $325$ codewords of the dual code $\mathcal{K}^\perp$ of weight $4$, which is the minimum dual weight. An example 
  is given by the codeword in $\mathbb{F}_2^{65}$ that has its four ones in coordinates $\{1,6,21,23\}$, i.e., the corresponding columns of $\mathcal{K}$ 
  sum up to the all zero vector. Let $\mathcal{T}$ be the set of $4$-subsets of $\{1,\dots,65\}$ that correspond to the $325$ codewords of the dual 
  code $\mathcal{K}^\perp$ of weight $4$. Using $\operatorname{ILP}_{16,44,\mathcal{K}}$ we can check (by prescribing) that no solution can satisfy 
  $\left(x_{1},x_{6},x_{21},x_{23}\right)\in\Big\{(0,0,0,1),(0,1,1,1),(1,1,1,1)\Big\}$. It can be computationally checked that the automorphism group 
  of $\mathcal{K}$ (of order $15600$) acts transitively on the set of $4$-tuples $\left(i_1,i_2,i_3,i_4\right)$ with $\left\{i_1,i_2,i_3,i_4\right\}\in\mathcal{T}$. 
  Thus, the conditions
  \begin{equation}
    \label{eq_even}
    \sum_{i\in T} x_i +2z_T= 2\quad\forall T\in\mathcal{T},
  \end{equation} 
  where $z_T\in\{0,1\}$ for all $T\in\mathcal{T}$, are satisfied for all integral solutions of $\operatorname{ILP}_{16,44,\mathcal{K}}$. We can check that the 
  code from the statement contains exactly $26$ codewords of weight $16$. Let $\mathcal{E}$ be the corresponding set of $15$-subsets of $\{1,\dots,65\}$ where 
  the codewords have a one. If $\sum_{i\in E} x_i=15$ for an $x\in\mathbb{F}_2^{65}$ with $\sum_{i=1}^{65}x_i=15$ and $E\in\mathcal{E}$, then $x$ is a solution of 
  $\operatorname{ILP}_{16,44,\mathcal{K}}$ that corresponds to the code $\mathcal{K}'$ from the statement. Thus we consider $\operatorname{ILP}_{16,44,\mathcal{K}}$ 
  with the additional constraints (\ref{eq_even}) and 
  \begin{equation}
    \sum_{i\in E} x_i\le 14
  \end{equation} 
  for all $E\in\mathcal{E}$. It turns out that no solution of that ILP exists so that we can conclude the statement. 
\end{proof}
Note that we do not impose that the automorphism group of $\mathcal{K}'$ contains the automorphism group $\operatorname{Aut}(\mathcal{K})$ of $\mathcal{K}$, when 
restricted to the first $65$ coordinates. However, the final solution has this property. In general, for $\pi\in \operatorname{Aut}(\mathcal{K})$ 
and $x$ a solution of $\operatorname{ILP}_{16,44,\mathcal{K}}$ we have that $\pi(x)$ is also a solution of $\operatorname{ILP}_{16,44,\mathcal{K}}$, 
which might correspond to either the same or a different code $\mathcal{K}'$. We remark that all ILP computations took just a few minutes. 

The unique possibility for $\mathcal{K}'$ can also be constructed as follows. Let $\mathcal{K}$ be the code of effective length $65$ in Theorem~\ref{thm_classification_1} 
and $\mathcal{D}$ be the code generated by the codewords of weight $4$ in $K^\perp$. It can be checked that $\mathcal{K}\le \mathcal{D}^\perp$ and $\dim(\mathcal{D}^\perp)=14$. 
Moreover $\mathcal{D}^\perp$ is partitioned by the cosets of $\mathcal{K}$ into sets of codewords of $\mathcal{D}^\perp$ whose weights are equivalent to either $0$, $1$, $2$, or 
$3$ modulo $4$. Taking the unique code $\overline{\mathcal{K}}$ of dimension $13$ with  $\mathcal{K}\le\overline{\mathcal{K}} \le\mathcal{D}^\perp$ whose codewords 
have weights that are either congruent to $0$ or $3$ modulo $4$ and adding a parity bit gives $\mathcal{K}'$.  
  
We remark that some parts of the computations in the proofs of Proposition~\ref{prop_max_weight_extended_44} and Theorem~\ref{thm_unique_extended_code} 
can be replaced by theoretical reasoning's. For example, if $\mathcal{K}'\backslash\mathcal{K}$ contains a codeword of weight $64$, then the corresponding residual 
code $\mathcal{R}$ in $\mathcal{K}'$ is a $2$-divisible linear code of effective length $n+1-64$, where $n$ is the effective length of $\mathcal{K}$. Since 
$\mathcal{K}$ and $\mathcal{K}'$ are projective, also $\mathcal{R}$ is projective. However, the smallest $2$-divisible projective binary linear code has 
length $3$, so that we obtain a contradiction. If we already know that $\sum_{i\in T} x_i\equiv 0\pmod 2$ for all $T\in\mathcal{T}$, see constraint~(\ref{eq_even}), 
then we can conclude that $\mathcal{K}'$ has to arise by adding a parity bit to $\overline{K}$ where $\mathcal{K}\le\overline{\mathcal{K}} \le\mathcal{D}^\perp$ 
and $\dim(\overline{\mathcal{K}})=13$. For nodal sextics it is of some interest that $\mathcal{K}'\backslash \mathcal{K}$ contains a codeword of weight $32$. Of course 
this directly follows from Theorem~\ref{thm_unique_extended_code}. However, we can also apply the first four MacWilliams identities 
together with $n=66$, $k=13$, $a_0=1$, $\sum_{i=1}^{15}a_i+\sum_{i>44}a_i+\sum_{i\,:\,i\not\equiv 0\pmod 4}a_i=0$, $a_{20}=0$, $a_{24}=390$, $a_{32}\ge 3055$, $a_{40}\ge 650$, $a_0^*=1$, 
$a_1^*=0$, $a_2^*=0$, and $a_3^*=0$ gives $a_{32}\ge 3535$, i.e., $\mathcal{K}'\backslash \mathcal{K}$ contains at least $480$ codewords of weight $32$.

Of course we can also apply the computational techniques of the proof of Proposition~\ref{prop_max_weight_extended_44} to $\codeeff{n}{11}{\{24,32,40\}}{2}$ or similar codes. It turns 
out that the unique $\codeeff{62}{11}{\{24,32,40\}}{2}$ code does not allow a code $\mathcal{K}'$ as specified in Proposition~\ref{prop_max_weight_extended_44}. The nine 
$\codeeff{63}{11}{\{24,32,40\}}{}2$ codes do not allow a code $\mathcal{K}'$ as specified in Proposition~\ref{prop_max_weight_extended_44} with maximum weight strictly larger than 
$44$ in $\mathcal{K}'\backslash\mathcal{K}$. However, for the $\codeeff{63}{11}{\{24,32,40\}}{2}$ code
$$
\mathcal{K}=\begin{pmatrix}
100001100110000011111100101110000001101100110011000010000000000\\
111100001110001011101011100000000001011100001101001001000000000\\
011100100010001010001001000110100101001000111011110100100000000\\
010100100110000000011000111010001100110100011101011100010000000\\
001101101000000100110011010000101100011100001111110000001000000\\
001011111001000001110111001011100100000010000100101100000100000\\
001000011000110100001110000110011011110011111100000000000010000\\
000111111000010011111110111110000111110001110100011100000001000\\
000000000111110000000001111110000000001111110011111100000000100\\
000000000000001111111111111110000000000000001111111100000000010\\
000000000000000000000000000001111111111111111111111100000000001\\
\end{pmatrix}
$$ 
with weight enumerator $W(z)=1z^0+310z^{24}+1551z^{32}+186z^{40}$ 
we can add the codeword $x\in\mathbb{F}_2^{64}$ of weight $36$ with
\begin{eqnarray*}
  \{i\,:\, x_i=1\}&=&\{
  2,6,9,10,17,19,22,27,28,30,33,34,35,36,39,41,42,44,\\
  &&45,46,48,49,50,51,53,54,55,56,57,58,59,60,61,62,63,64  
  \}
\end{eqnarray*}
to obtain $\mathcal{K}'$ such that $\mathcal{K}'\backslash\mathcal{K}$ has weight enumerator 
$W(z)=896z^{28}+1152^{36}$ or
the codeword $x\in\mathbb{F}_2^{64}$ of weight $40$ with
\begin{eqnarray*}
  \{i\,:\, x_i=1\}&=&\{
  4,6,8,9,12,14,15,21,23,25,26,28,29,30,31,33,37,38,39,41,\\
  &&43,44,45,46,47,48,49,50,51,52,53,55,56,57,58,59,60,62,63  
  \}
\end{eqnarray*}
to obtain $\mathcal{K}'$ such that $\mathcal{K}'\backslash\mathcal{K}$ has weight enumerator 
$W(z)=768z^{28}+384z^{32}+768^{36}+128z^{40}$. The other eight $\codeeff{63}{11}{\{24,32,40\}}{2}$ codes 
do not allow a code $\mathcal{K}'$, as specified in Proposition~\ref{prop_max_weight_extended_44}, 
at all. 
For the stated code $\mathcal{K}$ the possible maximum weights of $\mathcal{K}'\backslash\mathcal{K}$ are 
either $36$ or $40$. In both cases the minimum distance of $\mathcal{K}'\backslash\mathcal{K}$ is $28$, 
i.e., no codewords of weight $16$ can occur. (Computationally checked by minimizing 
$\sum_{c\in\mathcal{C}\backslash 0} y_c$ in the corresponding ILP, i.e., the minimum is attained at 
target value $2047$ in both cases.)

From the $47$ $\codeeff{64}{11}{\{24,32,40\}}{2}$ codes only the following two do allow a code $\mathcal{K}'$, as specified in Proposition~\ref{prop_max_weight_extended_44}: 
$$
\mathcal{K}=\begin{pmatrix}
0001010110111111100001010001100001010011010000011010010000000000\\
0100100000011101100100010001100001001110011011101001101000000000\\
1001100100010000101001010110000100010111011010111001000100000000\\
0101000000001110001100101110000100001110101000111011100010000000\\
1100011000100110011000000110000011011100000101010111100001000000\\
0001001001100010000100101101010000110101001101001111100000100000\\
0011000100011110000011101110111110001000100011000000100000010000\\
0011000011111110000000011101111110111011100000111111100000001000\\
0000111111111110000000000011110001111000011111111111100000000100\\
0000000000000001111111111111110000000111111111111111100000000010\\
0000000000000000000000000000001111111111111111111111100000000001\\
\end{pmatrix}
$$
and
$$
\mathcal{K}=\begin{pmatrix}
0000011010000001000110111110001110100101111000110010010000000000\\
0111001100100010000011101000001011100100101101101010001000000000\\
0000101110101011000110100111100101110100100101000000000100000000\\
0111101110001010001100101011001010000101010011010000000010000000\\
1001110100000110100000011010000101011011000100111001100001000000\\
1101010010011110010001000110111100000011001100001000100000100000\\
0101010001111110001011000001110100001000000011000111100000010000\\
0011001111111110000111000000010010111000111110111111100000001000\\
0000111111111110000000111111110001111000000001111111100000000100\\
0000000000000001111111111111110000000111111111111111100000000010\\
0000000000000000000000000000001111111111111111111111100000000001\\
\end{pmatrix}.
$$
Both codes have weight enumerator $W(z)=1z^0+246z^{24}+1551z^{32}+250z^{40}$, while the first group has an automorphism group of 
order $240$ and the second code has an automorphism group of order $5760$.
Again, the minimum weight in $\mathcal{K}'\backslash\mathcal{K}$ has to be $28$. In both cases we found an 
example with weight enumerator $W(z)=672z^{28}+352z^{32}+864z^{36}+160z^{40}$ of $\mathcal{K}'\backslash\mathcal{K}$ by adding an 
additional codeword of weight $40$:   
\begin{eqnarray*}
  \{i\,:\, x_i=1\}&=&\{
  2,7,11,12,14,16,18,21,22,27,29,30,31,33,34,35,37,39,40,41,\\
  &&42,43,45,46,47,48,50,51,52,55,56,57,58,59,60,61,62,63,64,65  
  \}
\end{eqnarray*}
and
\begin{eqnarray*}
  \{i\,:\, x_i=1\}&=&\{
  2,6,8,11,12,15,16,19,24,25,26,30,31,32,33,35,37,38,40,41,\\
  &&42,43,44,46,47,48,50,52,53,54,55,56,57,58,59,50,61,62,66,65  
  \}.
\end{eqnarray*}
In $\mathcal{K}'\backslash\mathcal{K}$ no maximum weight smaller than $40$ or larger than $44$ is possible. We remark that there do indeed 
exist $[\le 64,11,\{24,32,40,56\}]_2$ codes that contain a codeword of weight $56$. There counts per effective length are given by 
$62^1 63^{18} 64^{281}$.

\section{Classification of the optimal $\codeeff{n}{k}{24}{2}$ codes that are $8$-divisible}
\label{sec_class_optimal} 

In this appendix we list all $\codeeff{n}{k}{24}{2}$ codes that achieve the optimal minimal Hamming distance and are $8$-divisible. 
For each case we give a $n\times k$-generator matrix, the weight polynomial $W(z)$, and the order of the automorphism group of the corresponding 
multiset of points.  

\medskip

\noindent
$\codeeff{n}{k}{d}{q}=\codeeff{24}{1}{24}{2}$\\
111111111111111111111111\\
$W(z)=1z^{0}+1z^{24}$\\ 
$\# \operatorname{Aut}=1$

\medskip

\noindent
$\codeeff{n}{k}{d}{q}=\codeeff{36}{2}{24}{2}$\\
111111111111111111111110000000000010\\
111111111111000000000001111111111101\\
$W(z)=1z^{0}+3z^{24}$\\ 
$\# \operatorname{Aut}=6$

\medskip

\noindent
$\codeeff{n}{k}{d}{q}=\codeeff{42}{3}{24}{2}$\\ 
111111111111111111111110000000000000000100\\
111111111111000000000001111111111100000010\\
111111000000111111000001111110000011111001\\
$1W(z)=z^{0}+7z^{24}$\\
$\# \operatorname{Aut}=168$

\medskip

\noindent
$\codeeff{n}{k}{d}{q}=\codeeff{44}{3}{24}{2}$\\ 
11111111111111111111111000000000000000000100\\
11111111111100000000000111111111110000000010\\
11110000000011110000000111111110001111111001\\
$W(z)=1z^{0}+6z^{24}+1z^{32}$\\
$\# \operatorname{Aut}=24$

\medskip

\noindent
$\codeeff{n}{k}{d}{q}=\codeeff{45}{4}{24}{2}$\\
111111111111111111111110000000000000000001000\\
111111111111000000000001111111111100000000100\\
111111000000111111000001111110000011111000010\\
111000111000111000111001110001110011100110001\\
$W(z)=1z^{0}+15z^{24}$\\
$\# \operatorname{Aut}=20160$\\

\medskip

\noindent
$\codeeff{n}{k}{d}{q}=\codeeff{46}{4}{24}{2}$\\
1111111111111111111111100000000000000000001000\\
1111111111110000000000011111111111000000000100\\
1111110000001111110000011111100000111110000010\\
1100001100001100001100011110011110111101110001\\
$W(z)=1z^{0}+14z^{24}+1z^{32}$\\
$\# \operatorname{Aut}=1344$

\medskip

\noindent
$\codeeff{n}{k}{d}{q}=\codeeff{47}{4}{24}{2}$\\
11111111111111111111111000000000000000000001000\\
11111111111100000000000111111111110000000000100\\
11111100000011111100000111111000001111100000010\\
10000010000010000010000111110111111111111110001\\
$W(z)=1z^{0}+14z^{24}+1z^{40}$\\
$\# \operatorname{Aut}=1344$

\medskip

\noindent
$\codeeff{n}{k}{d}{q}=\codeeff{47}{4}{24}{2}$\\
11111111111111111111111000000000000000000001000\\
11111111111100000000000111111111110000000000100\\
11111100000011111100000111111000001111100000010\\
10000010000011100011100111000111001111111110001\\
$W(z)=1z^{0}+13z^{24}+2z^{32}$\\
$\# \operatorname{Aut}=192$

\medskip

\noindent
$\codeeff{n}{k}{d}{q}=\codeeff{48}{4}{24}{2}$\\
111111111111111111111110000000000000000000001000\\
111111111111000000000001111111111100000000000100\\
111111000000111111000001111110000011111000000010\\
000000110000110000000001111111111011110111110001\\
$W(z)=1z^{0}+13z^{24}+1z^{32}+1z^{40}$\\
$\# \operatorname{Aut}=96$

\medskip

\noindent
$\codeeff{n}{k}{d}{q}=\codeeff{48}{4}{24}{2}$\\
111111111111111111111110000000000000000000001000\\
111111111111000000000001111111111100000000000100\\
111111000000111111000001111110000011111000000010\\
000000110000110000111101100001111011110111110001\\
$W(z)=1z^{0}+12z^{24}+3z^{32}$\\
$\# \operatorname{Aut}=48$

\medskip

\noindent
$\codeeff{n}{k}{d}{q}=\codeeff{48}{4}{24}{2}$\\
111111111111111111111110000000000000000000001000\\
111111111111000000000001111111111100000000000100\\
111111000000111111000001111110000011111000000010\\
000000111111111111000001111110000000000111110001\\
$W(z)=1z^{0}+14z^{24}+1z^{48}$\\
$\# \operatorname{Aut}=1344$

\medskip

\noindent
$\codeeff{n}{k}{d}{q}=\codeeff{48}{4}{24}{2}$\\
111111111111111111111110000000000000000000001000\\
111111111111000000000001111111111100000000000100\\
111100000000111100000001111111100011111110000010\\
000011110000111111110001111000000011110001110001\\
$W(z)=1z^{0}+12z^{24}+3z^{32}$\\
$\# \operatorname{Aut}=576$

\medskip

\noindent
$\codeeff{n}{k}{d}{q}=\codeeff{47}{5}{24}{2}$\\
11111111111111111111111000000000000000000010000\\
11111111111100000000000111111111110000000001000\\
11111100000011111100000111111000001111100000100\\
11100011100011100011100111000111001110011000010\\
10010010010010010010010110110110111101111100001\\
$W(z)=1z^{0}+30z^{24}+1z^{32}$\\
$\# \operatorname{Aut}=322560$

\medskip

\noindent
$\codeeff{n}{k}{d}{q}=\codeeff{48}{5}{24}{2}$\\
111111111111111111111110000000000000000000010000\\
111111111111000000000001111111111100000000001000\\
111111000000111111000001111110000011111000000100\\
111000111000111000111001110001110011100110000010\\
000100100000100000000101101111111111111111100001\\
$W(z)=1z^{0}+29z^{24}+1z^{32}+1z^{40}$\\
$\# \operatorname{Aut}=10752$

\medskip

\noindent
$\codeeff{n}{k}{d}{q}=\codeeff{48}{5}{24}{2}$\\
111111111111111111111110000000000000000000010000\\
111111111111000000000001111111111100000000001000\\
111111000000111111000001111110000011111000000100\\
111000111000111000111001110001110011100110000010\\
000100100000100110110101101001001111111111100001\\
$W(z)=1z^{0}+28z^{24}+3z^{32}$\\
$\# \operatorname{Aut}=2304$

\medskip

\noindent
$\codeeff{n}{k}{d}{q}=\codeeff{48}{5}{24}{2}$\\
111111111111111111111110000000000000000000010000\\
111111111111000000000001111111111100000000001000\\
111111000000111111000001111110000011111000000100\\
111000111000111000111001110001110011100110000010\\
111000000111000111111000001111110011100001100001\\
$W(z)=1z^{0}+30z^{24}+1z^{48}$\\
$\# \operatorname{Aut}=322560$

\medskip

\noindent
$\codeeff{n}{k}{d}{q}=\codeeff{48}{5}{24}{2}$\\
111111111111111111111110000000000000000000010000\\
111111111111000000000001111111111100000000001000\\
111111000000111111000001111110000011111000000100\\
110000110000110000110001111001111011110111000010\\
001100111100111100001101100111100011000110100001\\
$W(z)=1z^{0}+28z^{24}+3z^{32}$\\
$\# \operatorname{Aut}=64512$

\medskip

\noindent
$\codeeff{n}{k}{d}{q}=\codeeff{49}{5}{24}{2}$\\
1111111111111111111111100000000000000000000010000\\
1111111111110000000000011111111111000000000001000\\
1111110000001111110000011111100000111110000000100\\
0000001100001100001111011000011110111101111100010\\
1100000011000011001100010111011101111011110000001\\
$W(z)=1z^{0}+26z^{24}+5z^{32}$\\
$\# \operatorname{Aut}=120$

\medskip

\noindent
$\codeeff{n}{k}{d}{q}=\codeeff{50}{5}{24}{2}$\\
11111111111111111111111000000000000000000000010000\\
11111111111100000000000111111111110000000000001000\\
11111100000011111100000111111000001111100000000100\\
00000011000011000000000111111111101111011111000010\\
10000000111011111011100111000100001110011100100001\\
$W(z)=1z^{0}+25z^{24}+5z^{32}+1z^{40}$\\
$\# \operatorname{Aut}=120$

\medskip

\noindent
$\codeeff{n}{k}{d}{q}=\codeeff{50}{5}{24}{2}$\\
11111111111111111111111000000000000000000000010000\\
11111111111100000000000111111111110000000000001000\\
11111100000011111100000111111000001111100000000100\\
00000011000011000000000111111111101111011111000010\\
11000000110011110011110110000110001111011000100001\\
$W(z)=1z^{0}+25z^{24}+5z^{32}+1z^{40}$\\
$\# \operatorname{Aut}=120$

\medskip

\noindent
$\codeeff{n}{k}{d}{q}=\codeeff{50}{5}{24}{2}$\\
11111111111111111111111000000000000000000000010000\\
11111111111100000000000111111111110000000000001000\\
11111100000011111100000111111000001111100000000100\\
00000011000011000011110110000111101111011111000010\\
10000000111010110010000111110110011110011100100001\\
$W(z)=1z^{0}+24z^{24}+7z^{32}$\\
$\# \operatorname{Aut}=72$

\medskip

\noindent
$\codeeff{n}{k}{d}{q}=\codeeff{50}{5}{24}{2}$\\
11111111111111111111111000000000000000000000010000\\
11111111111100000000000111111111110000000000001000\\
11111100000011111100000111111000001111100000000100\\
00000011000011000011110110000111101111011111000010\\
11000000110010100010001101110111011111011000100001\\
$W(z)=1z^{0}+24z^{24}+7z^{32}$\\
$\# \operatorname{Aut}=48$

\medskip

\noindent
$\codeeff{n}{k}{d}{q}=\codeeff{50}{5}{24}{2}$\\
11111111111111111111111000000000000000000000010000\\
11111111111100000000000111111111110000000000001000\\
11111100000011111100000111111000001111100000000100\\
00000011000011000011110110000111101111011111000010\\
11000000110000110011000111111110001100011110100001\\
$W(z)=1z^{0}+24z^{24}+7z^{32}$\\
$\# \operatorname{Aut}=144$

\medskip

\noindent
$\codeeff{n}{k}{d}{q}=\codeeff{50}{5}{24}{2}$\\
11111111111111111111111000000000000000000000010000\\
11111111111100000000000111111111110000000000001000\\
11111100000011111100000111111000001111100000000100\\
00000011000011000011110110000111101111011111000010\\
11100000100000100011001101111111001110011100100001\\
$W(z)=1z^{0}+24z^{24}+7z^{32}$\\
$\# \operatorname{Aut}=48$

\medskip

\noindent
$\codeeff{n}{k}{d}{q}=\codeeff{48}{6}{24}{2}$\\
111111111111111111111110000000000000000000100000\\
111111111111000000000001111111111100000000010000\\
111111000000111111000001111110000011111000001000\\
111000111000111000111001110001110011100110000100\\
100100100100100100100101101101101111011111000010\\
010110110010110010010111001011011010110101000001\\
$W(z)=1z^{0}+60z^{24}+3z^{32}$\\
$\# \operatorname{Aut}=30965760$

\medskip

\noindent
$\codeeff{n}{k}{d}{q}=\codeeff{49}{6}{24}{2}$\\
1111111111111111111111100000000000000000000100000\\
1111111111110000000000011111111111000000000010000\\
1111110000001111110000011111100000111110000001000\\
1110001110001110001110011100011100111001100000100\\
0001001000001001101101011010010011111111111000010\\
1000101101100101001011100011011110110001011000001\\
$W(z)=1z^{0}+56z^{24}+7z^{32}$\\
$\# \operatorname{Aut}=5040$

\medskip

\noindent
$\codeeff{n}{k}{d}{q}=\codeeff{50}{6}{24}{2}$\\
11111111111111111111111000000000000000000000100000\\
11111111111100000000000111111111110000000000010000\\
11111100000011111100000111111000001111100000001000\\
11000011000011000011000111100111101111011100000100\\
00110011110011110000110110011110001100011010000010\\
00000000101000101011110001010111011110111111000001\\
$W(z)=1z^{0}+52z^{24}+11z^{32}$\\
$\# \operatorname{Aut}=192$

\medskip

\noindent
$\codeeff{n}{k}{d}{q}=\codeeff{50}{6}{24}{2}$\\
11111111111111111111111000000000000000000000100000\\
11111111111100000000000111111111110000000000010000\\
11111100000011111100000111111000001111100000001000\\
00000011000011000011110110000111101111011111000100\\
11000000110000110011000101110111011110111100000010\\
00110000101010101111101001001110001110110011000001\\
$W(z)=1z^{0}+52z^{24}+11z^{32}$\\
$\# \operatorname{Aut}=120$

\medskip

\noindent
$\codeeff{n}{k}{d}{q}=\codeeff{51}{6}{24}{2}$\\
111111111111111111111110000000000000000000000100000\\
111111111111000000000001111111111100000000000010000\\
111111000000111111000001111110000011111000000001000\\
000000110000110000111101100001111011110111110000100\\
100000001110101100100001111101100111100111001000010\\
010000001101010011010001111011010111010100111000001\\
$W(z)=1z^{0}+48z^{24}+15z^{32}$\\
$\# \operatorname{Aut}=96$

\medskip

\noindent
$\codeeff{n}{k}{d}{q}=\codeeff{51}{6}{24}{2}$\\
111111111111111111111110000000000000000000000100000\\
111111111111000000000001111111111100000000000010000\\
111111000000111111000001111110000011111000000001000\\
000000110000110000111101100001111011110111110000100\\
100000001110101100100001111101100111100111001000010\\
011000001001000011110001011011011111101110100000001\\
$W(z)=1z^{0}+48z^{24}+15z^{32}$\\
$\# \operatorname{Aut}=12$

\medskip

\noindent
$\codeeff{n}{k}{d}{q}=\codeeff{51}{6}{24}{2}$\\
111111111111111111111110000000000000000000000100000\\
111111111111000000000001111111111100000000000010000\\
111111000000111111000001111110000011111000000001000\\
000000110000110000111101100001111011110111110000100\\
100000001110101100100001111101100111100111001000010\\
011000001100010010010011011011110111110000111000001\\
$W(z)=1z^{0}+48z^{24}+15z^{32}$\\
$\# \operatorname{Aut}=12$

\medskip

\noindent
$\codeeff{n}{k}{d}{q}=\codeeff{51}{6}{24}{2}$\\
111111111111111111111110000000000000000000000100000\\
111111111111000000000001111111111100000000000010000\\
111111000000111111000001111110000011111000000001000\\
000000110000110000111101100001111011110111110000100\\
110000001100101000100011011101110111110110001000010\\
101000100010011110000001111001110110001101111000001\\
$W(z)=1z^{0}+48z^{24}+15z^{32}$\\
$\# \operatorname{Aut}=720$

\medskip

\noindent
$\codeeff{n}{k}{d}{q}=\codeeff{51}{6}{24}{2}$\\
111111111111111111111110000000000000000000000100000\\
111111111111000000000001111111111100000000000010000\\
111111000000111111000001111110000011111000000001000\\
000000110000110000111101100001111011110111110000100\\
110000001100001100110001111111100011000111101000010\\
001100000011101110111010011001100011101110010000001\\
$W(z)=1z^{0}+48z^{24}+15z^{32}$\\
$\# \operatorname{Aut}=360$

\medskip

\noindent
$\codeeff{n}{k}{d}{q}=\codeeff{50}{7}{24}{2}$\\
11111111111111111111111000000000000000000001000000\\
11111111111100000000000111111111110000000000100000\\
11111100000011111100000111111000001111100000010000\\
11100011100011100011100111000111001110011000001000\\
00010010000010011011010110100100111111111110000100\\
10001011011001010010111000110111101100010110000010\\
01001010110100101001111100001011111001011100000001\\
$W(z)=1z^{0}+108z^{24}+19z^{32}$\\
$\# \operatorname{Aut}=5760$

\medskip

\noindent
$\codeeff{n}{k}{d}{q}=\codeeff{51}{7}{24}{2}$\\
111111111111111111111110000000000000000000001000000\\
111111111111000000000001111111111100000000000100000\\
111111000000111111000001111110000011111000000010000\\
110000110000110000110001111001111011110111000001000\\
001100111100111100001101100111100011000110100000100\\
000000001010001010111100010101110111101111110000010\\
000011100100100100001101001111011110111101010000001\\
$W(z)=1z^{0}+100z^{24}+27z^{32}$\\
$\# \operatorname{Aut}=240$

\medskip

\noindent
$\codeeff{n}{k}{d}{q}=\codeeff{52}{7}{24}{2}$\\
1111111111111111111111100000000000000000000001000000\\
1111111111110000000000011111111111000000000000100000\\
1111110000001111110000011111100000111110000000010000\\
0000001100001100001111011000011110111101111100001000\\
1000000011101011001000011111011001111001110010000100\\
0100000011010100110100011110110101110101001110000010\\
0011001000011111001011110111000000001010111110000001\\
$W(z)=1z^{0}+92z^{24}+35z^{32}$\\
$\# \operatorname{Aut}=16$

\medskip

\noindent
$\codeeff{n}{k}{d}{q}=\codeeff{52}{7}{24}{2}$\\
1111111111111111111111100000000000000000000001000000\\
1111111111110000000000011111111111000000000000100000\\
1111110000001111110000011111100000111110000000010000\\
0000001100001100001111011000011110111101111100001000\\
1000000011101011001000011111011001111001110010000100\\
0110000010010000111100010110110111111011101000000010\\
1111000000000110000010101110101111100101111010000001\\
$W(z)=1z^{0}+92z^{24}+35z^{32}$\\
$\# \operatorname{Aut}=32$

\medskip

\noindent
$\codeeff{n}{k}{d}{q}=\codeeff{52}{7}{24}{2}$\\
1111111111111111111111100000000000000000000001000000\\
1111111111110000000000011111111111000000000000100000\\
1111110000001111110000011111100000111110000000010000\\
0000001100001100001111011000011110111101111100001000\\
1000000011101011001000011111011001111001110010000100\\
0110000010010000111100010110110111111011101000000010\\
0001100011001011101110100101000110111010011100000001\\
$W(z)=1z^{0}+92z^{24}+35z^{32}$\\
$\# \operatorname{Aut}=168$

\medskip

\noindent
$\codeeff{n}{k}{d}{q}=\codeeff{52}{7}{24}{2}$\\
1111111111111111111111100000000000000000000001000000\\
1111111111110000000000011111111111000000000000100000\\
1111110000001111110000011111100000111110000000010000\\
0000001100001100001111011000011110111101111100001000\\
1000000011101011001000011111011001111001110010000100\\
0110000010010000111100010110110111111011101000000010\\
0101001010000011111110101000111000111010010110000001\\
$W(z)=1z^{0}+92z^{24}+35z^{32}$\\
$\# \operatorname{Aut}=42$

\medskip

\noindent
$\codeeff{n}{k}{d}{q}=\codeeff{53}{7}{24}{2}$\\
11111111111111111111111000000000000000000000001000000\\
11111111111100000000000111111111110000000000000100000\\
11111100000011111100000111111000001111100000000010000\\
00000011000011000011110110000111101111011111000001000\\
11000000110000110011000111111110001100011110100000100\\
00000000001110101110001111100110001011111001110000010\\
00100000101110011000100111110001111110000111100000001\\
$W(z)=1z^{0}+84z^{24}+43z^{32}$\\
$\# \operatorname{Aut}=3$

\medskip

\noindent
$\codeeff{n}{k}{d}{q}=\codeeff{53}{7}{24}{2}$\\
11111111111111111111111000000000000000000000001000000\\
11111111111100000000000111111111110000000000000100000\\
11111100000011111100000111111000001111100000000010000\\
00000011000011000011110110000111101111011111000001000\\
11000000110000110011000111111110001100011110100000100\\
00000000001110101110001111100110001011111001110000010\\
11000010100001101111110001010100011011110100010000001\\
$W(z)=1z^{0}+84z^{24}+43z^{32}$\\
$\# \operatorname{Aut}=168$

\medskip

\noindent
$\codeeff{n}{k}{d}{q}=\codeeff{53}{7}{24}{2}$\\
11111111111111111111111000000000000000000000001000000\\
11111111111100000000000111111111110000000000000100000\\
11111100000011111100000111111000001111100000000010000\\
00000011000011000011110110000111101111011111000001000\\
11000000110000110011000111111110001100011110100000100\\
00000000101011101000110101110100011011111100110000010\\
00110000001100010110100011110111010111111010000000001\\
$W(z)=1z^{0}+84z^{24}+43z^{32}$\\
$\# \operatorname{Aut}=4$

\medskip

\noindent
$\codeeff{n}{k}{d}{q}=\codeeff{53}{7}{24}{2}$\\
11111111111111111111111000000000000000000000001000000\\
11111111111100000000000111111111110000000000000100000\\
11111100000011111100000111111000001111100000000010000\\
00000011000011000011110110000111101111011111000001000\\
11000000110000110011000111111110001100011110100000100\\
00000000101011101000110101110100011011111100110000010\\
10110000001000000110101111110001111100111101000000001\\
$W(z)=1z^{0}+84z^{24}+43z^{32}$\\
$\# \operatorname{Aut}=12$

\medskip

\noindent
$\codeeff{n}{k}{d}{q}=\codeeff{53}{7}{24}{2}$\\
11111111111111111111111000000000000000000000001000000\\
11111111111100000000000111111111110000000000000100000\\
11111100000011111100000111111000001111100000000010000\\
00000011000011000011110110000111101111011111000001000\\
11000000110000110011000111111110001100011110100000100\\
00000000101011101000110101110100011011111100110000010\\
00111010010100100000100001110111001111110011100000001\\
$W(z)=1z^{0}+84z^{24}+43z^{32}$\\
$\# \operatorname{Aut}=3$

\medskip

\noindent
$\codeeff{n}{k}{d}{q}=\codeeff{53}{7}{24}{2}$\\
11111111111111111111111000000000000000000000001000000\\
11111111111100000000000111111111110000000000000100000\\
11111100000011111100000111111000001111100000000010000\\
00000011000011000011110110000111101111011111000001000\\
11000000110000110011000111111110001100011110100000100\\
00000000101011101000110101110100011011111100110000010\\
11110010001000000000101011110001101110111101100000001\\
$W(z)=1z^{0}+84z^{24}+43z^{32}$\\
$\# \operatorname{Aut}=6$

\medskip

\noindent
$\codeeff{n}{k}{d}{q}=\codeeff{53}{7}{24}{2}$\\
11111111111111111111111000000000000000000000001000000\\
11111111111100000000000111111111110000000000000100000\\
11111100000011111100000111111000001111100000000010000\\
00000011000011000011110110000111101111011111000001000\\
11000000110000110011000111111110001100011110100000100\\
00000000101011101000110101110100011011111100110000010\\
11000000010110101111101001001101001011111001000000001\\
$W(z)=1z^{0}+84z^{24}+43z^{32}$\\
$\# \operatorname{Aut}=42$

\medskip

\noindent
$\codeeff{n}{k}{d}{q}=\codeeff{54}{7}{24}{2}$\\
111111111111111111111110000000000000000000000001000000\\
111111111111000000000001111111111100000000000000100000\\
111111000000111111000001111110000011111000000000010000\\
000000110000110000000001111111111011110111110000001000\\
110000001100111100111101100001100011110110001000000100\\
001100000011110011111101111000000011000111100100000010\\
000000100010101000111101000101011111101001101110000001\\
$W(z)=1z^{0}+78z^{24}+47z^{32}+2z^{40}$\\
$\# \operatorname{Aut}=12$

\medskip

\noindent
$\codeeff{n}{k}{d}{q}=\codeeff{54}{7}{24}{2}$\\
111111111111111111111110000000000000000000000001000000\\
111111111111000000000001111111111100000000000000100000\\
111111000000111111000001111110000011111000000000010000\\
000000110000110000000001111111111011110111110000001000\\
110000001100111100111101100001100011110110001000000100\\
001100000011110011111101111000000011000111100100000010\\
000011001100111111110000011000011000110111100010000001\\
$W(z)=1z^{0}+82z^{24}+39z^{32}+6z^{40}$\\
$\# \operatorname{Aut}=240$

\medskip

\noindent
$\codeeff{n}{k}{d}{q}=\codeeff{54}{7}{24}{2}$\\
111111111111111111111110000000000000000000000001000000\\
111111111111000000000001111111111100000000000000100000\\
111111000000111111000001111110000011111000000000010000\\
000000110000110000000001111111111011110111110000001000\\
110000001100111100111101100001100011110110001000000100\\
001100000011110011111101111000000011000111100100000010\\
001110001000111011110010000101011011100110010010000001\\
$W(z)=1z^{0}+78z^{24}+47z^{32}+2z^{40}$\\

$\# \operatorname{Aut}=16$

\medskip

\noindent
$\codeeff{n}{k}{d}{q}=\codeeff{54}{7}{24}{2}$\\
111111111111111111111110000000000000000000000001000000\\
111111111111000000000001111111111100000000000000100000\\
111111000000111111000001111110000011111000000000010000\\
000000110000110000000001111111111011110111110000001000\\
110000001100111100111101100001100011110110001000000100\\
001100000011111111110000011001100011000111100100000010\\
101000000011110011101111010000011011110100010010000001\\
$W(z)=1z^{0}+78z^{24}+47z^{32}+2z^{40}$\\
$\# \operatorname{Aut}=4$

\medskip

\noindent
$\codeeff{n}{k}{d}{q}=\codeeff{54}{7}{24}{2}$\\
111111111111111111111110000000000000000000000001000000\\
111111111111000000000001111111111100000000000000100000\\
111100000000111100000001111111100011111110000000010000\\
000011110000111111110001111000000011110001110000001000\\
000000000000000011001101111110011011111101101110000100\\
000010001110110000100011110111010110001110001110000010\\
100001101000101100010001111000111101001011001110000001\\
$W(z)=1z^{0}+77z^{24}+49z^{32}+1z^{40}$\\
$\# \operatorname{Aut}=2$

\medskip

\noindent
$\codeeff{n}{k}{d}{q}=\codeeff{54}{7}{24}{2}$\\
111111111111111111111110000000000000000000000001000000\\
111111111111000000000001111111111100000000000000100000\\
111100000000111100000001111111100011111110000000010000\\
000011110000111111110001111000000011110001110000001000\\
000000000000000011001101111110011011111101101110000100\\
000010001110110000100011110111010110001110001110000010\\
110001001101000000100011000110101101111111101000000001\\
$W(z)=1z^{0}+77z^{24}+49z^{32}+1z^{40}$\\
$\# \operatorname{Aut}=6$

\medskip

\noindent
$\codeeff{n}{k}{d}{q}=\codeeff{54}{7}{24}{2}$\\
111111111111111111111110000000000000000000000001000000\\
111111111111000000000001111111111100000000000000100000\\
111100000000111100000001111111100011111110000000010000\\
000011110000111111110001111000000011110001110000001000\\
000000000000000011001101111110011011111101101110000100\\
000010001110110000100011110111010110001110001110000010\\
110010001101000000010010001110101111101111101000000001\\
$W(z)=1z^{0}+78z^{24}+47z^{32}+2z^{40}$\\
$\# \operatorname{Aut}=2$

\medskip

\noindent
$\codeeff{n}{k}{d}{q}=\codeeff{54}{7}{24}{2}$\\
111111111111111111111110000000000000000000000001000000\\
111111111111000000000001111111111100000000000000100000\\
111100000000111100000001111111100011111110000000010000\\
000011110000111111110001111000000011110001110000001000\\
000000000000000011001101111110011011111101101110000100\\
000010001110110000100011110111010110001110001110000010\\
110011110000101000000001100100110101101011101110000001\\
$W(z)=1z^{0}+77z^{24}+49z^{32}+1z^{40}$\\
$\# \operatorname{Aut}=2$

\medskip

\noindent
$\codeeff{n}{k}{d}{q}=\codeeff{54}{7}{24}{2}$\\
111111111111111111111110000000000000000000000001000000\\
111111111111000000000001111111111100000000000000100000\\
111100000000111100000001111111100011111110000000010000\\
000011110000111111110001111000000011110001110000001000\\
000000000000000011001101111110011011111101101110000100\\
000010001110110000100011110111010110001110001110000010\\
100000001110101110111011100100001001000001111110000001\\
$W(z)=1z^{0}+77z^{24}+49z^{32}+1z^{40}$\\
$\# \operatorname{Aut}=4$

\medskip

\noindent
$\codeeff{n}{k}{d}{q}=\codeeff{54}{7}{24}{2}$\\
111111111111111111111110000000000000000000000001000000\\
111111111111000000000001111111111100000000000000100000\\
111100000000111100000001111111100011111110000000010000\\
000011110000111111110001111000000011110001110000001000\\
000000000000000011001101111110011011111101101110000100\\
000010001110110000100011110111010110001110001110000010\\
100000001110101110111011100100001001111100011000000001\\
$W(z)=1z^{0}+77z^{24}+49z^{32}+1z^{40}$\\
$\# \operatorname{Aut}=6$

\medskip

\noindent
$\codeeff{n}{k}{d}{q}=\codeeff{54}{7}{24}{2}$\\
111111111111111111111110000000000000000000000001000000\\
111111111111000000000001111111111100000000000000100000\\
111100000000111100000001111111100011111110000000010000\\
000011110000111111110001111000000011110001110000001000\\
000000000000000011001101111110011011111101101110000100\\
000010001110110000100011110111010110001110001110000010\\
000001101111110010011011100110011001100001011000000001\\
$W(z)=1z^{0}+77z^{24}+49z^{32}+1z^{40}$\\
$\# \operatorname{Aut}=4$

\medskip

\noindent
$\codeeff{n}{k}{d}{q}=\codeeff{54}{7}{24}{2}$\\
111111111111111111111110000000000000000000000001000000\\
111111111111000000000001111111111100000000000000100000\\
111100000000111100000001111111100011111110000000010000\\
000011110000111111110001111000000011110001110000001000\\
000000000000000011001101111110011011111101101110000100\\
000010001110110000100011110111010110001110001110000010\\
111001101110101111011101111000111101111011001000000001\\
$W(z)=1z^{0}+78z^{24}+47z^{32}+2z^{40}$\\
$\# \operatorname{Aut}=16$

\medskip

\noindent
$\codeeff{n}{k}{d}{q}=\codeeff{54}{7}{24}{2}$\\
111111111111111111111110000000000000000000000001000000\\
111111111111000000000001111111111100000000000000100000\\
111100000000111100000001111111100011111110000000010000\\
000011110000111111110001111000000011110001110000001000\\
000000000000000011001101111110011011111101101110000100\\
000011001100110000110001111101010111001010001110000010\\
110000001111001100000001100100101111000111101110000001\\
$W(z)=1z^{0}+79z^{24}+45z^{32}+3z^{40}$\\
$\# \operatorname{Aut}=4$

\medskip

\noindent
$\codeeff{n}{k}{d}{q}=\codeeff{54}{7}{24}{2}$\\
111111111111111111111110000000000000000000000001000000\\
111111111111000000000001111111111100000000000000100000\\
111100000000111100000001111111100011111110000000010000\\
000011110000111111110001111000000011110001110000001000\\
100010000000100000001001110000011111111111111110000100\\
000010001110111111100101000100011010001001101110000010\\
010011001110100011010110000010000101111101111000000001\\
$W(z)=1z^{0}+77z^{24}+49z^{32}+1z^{40}$\\
$\# \operatorname{Aut}=2$

\medskip

\noindent
$\codeeff{n}{k}{d}{q}=\codeeff{54}{7}{24}{2}$\\
111111111111111111111110000000000000000000000001000000\\
111111111111000000000001111111111100000000000000100000\\
111100000000111100000001111111100011111110000000010000\\
000011110000111111110001111000000011110001110000001000\\
100010000000100000001001110000011111111111111110000100\\
000010001110111111100101000100011010001001101110000010\\
111000001101010011011100000010000110000111111110000001\\
$W(z)=1z^{0}+78z^{24}+47z^{32}+2z^{40}$\\
$\# \operatorname{Aut}=6$

\medskip

\noindent
$\codeeff{n}{k}{d}{q}=\codeeff{54}{7}{24}{2}$\\
111111111111111111111110000000000000000000000001000000\\
111111111111000000000001111111111100000000000000100000\\
111100000000111100000001111111100011111110000000010000\\
000011110000111111110001111000000011110001110000001000\\
000010001000000011101111000111011010001111101110000100\\
000000000110110011001101111100100001101101011110000010\\
000001000111001110001001110111011011011000111000000001\\
$W(z)=1z^{0}+76z^{24}+51z^{32}$\\
$\# \operatorname{Aut}=4$

\medskip

\noindent
$\codeeff{n}{k}{d}{q}=\codeeff{54}{7}{24}{2}$\\
111111111111111111111110000000000000000000000001000000\\
111111111111000000000001111111111100000000000000100000\\
111100000000111100000001111111100011111110000000010000\\
000011110000111111110001111000000011110001110000001000\\
000010001000000011101111000111011010001111101110000100\\
000000000110110011001101111100100001101101011110000010\\
100011100000001011000010111010111101101111001000000001\\
$W(z)=1z^{0}+76z^{24}+51z^{32}$\\
$\# \operatorname{Aut}=1$

\medskip

\noindent
$\codeeff{n}{k}{d}{q}=\codeeff{54}{7}{24}{2}$\\
111111111111111111111110000000000000000000000001000000\\
111111111111000000000001111111111100000000000000100000\\
111100000000111100000001111111100011111110000000010000\\
000011110000111111110001111000000011110001110000001000\\
000010001000000011101111000111011010001111101110000100\\
000000000110110011001101111100100001101101011110000010\\
100001101000100000101011111110100111010000101110000001\\
$W(z)=1z^{0}+76z^{24}+51z^{32}$\\
$\# \operatorname{Aut}=2$

\medskip

\noindent
$\codeeff{n}{k}{d}{q}=\codeeff{54}{7}{24}{2}$\\
111111111111111111111110000000000000000000000001000000\\
111111111111000000000001111111111100000000000000100000\\
111100000000111100000001111111100011111110000000010000\\
000011110000111111110001111000000011110001110000001000\\
000010001000000011101111000111011010001111101110000100\\
000000000110110011001101111100100001101101011110000010\\
100001100100001010110001111011110010001010011110000001\\
$W(z)=1z^{0}+76z^{24}+51z^{32}$\\
$\# \operatorname{Aut}=2$

\medskip

\noindent
$\codeeff{n}{k}{d}{q}=\codeeff{54}{7}{24}{2}$\\
111111111111111111111110000000000000000000000001000000\\
111111111111000000000001111111111100000000000000100000\\
111100000000111100000001111111100011111110000000010000\\
000011110000111111110001111000000011110001110000001000\\
000010001000000011101111000111011010001111101110000100\\
000000000110110011001101111100100001101101011110000010\\
100011100101100000000010100111100101101110101110000001\\
$W(z)=1z^{0}+76z^{24}+51z^{32}$\\
$\# \operatorname{Aut}=4$

\medskip

\noindent
$\codeeff{n}{k}{d}{q}=\codeeff{54}{7}{24}{2}$\\
111111111111111111111110000000000000000000000001000000\\
111111111111000000000001111111111100000000000000100000\\
111100000000111100000001111111100011111110000000010000\\
000011110000111111110001111000000011110001110000001000\\
000010001000000011101111000111011010001111101110000100\\
000000000110110011001101111100100001101101011110000010\\
100001111100100000000011110100110010011111111000000001\\
$W(z)=1z^{0}+76z^{24}+51z^{32}$\\
$\# \operatorname{Aut}=6$

\medskip

\noindent
$\codeeff{n}{k}{d}{q}=\codeeff{54}{7}{24}{2}$\\
111111111111111111111110000000000000000000000001000000\\
111111111111000000000001111111111100000000000000100000\\
111100000000111100000001111111100011111110000000010000\\
000011110000111111110001111000000011110001110000001000\\
000010001000000011101111000111011010001111101110000100\\
000000000110110011001101111100100001101101011110000010\\
110001001110000000010011110011100000011110111110000001\\
$W(z)=1z^{0}+77z^{24}+49z^{32}+1z^{40}$\\
$\# \operatorname{Aut}=2$

\medskip

\noindent
$\codeeff{n}{k}{d}{q}=\codeeff{54}{7}{24}{2}$\\
111111111111111111111110000000000000000000000001000000\\
111111111111000000000001111111111100000000000000100000\\
111100000000111100000001111111100011111110000000010000\\
000011110000111111110001111000000011110001110000001000\\
000010001000000011101111000111011010001111101110000100\\
000000000110110011001101111100100001101101011110000010\\
110001001000110011011111000011100001110010111000000001\\
$W(z)=1z^{0}+77z^{24}+49z^{32}+1z^{40}$\\
$\# \operatorname{Aut}=6$

\medskip

\noindent
$\codeeff{n}{k}{d}{q}=\codeeff{54}{7}{24}{2}$\\
111111111111111111111110000000000000000000000001000000\\
111111111111000000000001111111111100000000000000100000\\
111100000000111100000001111111100011111110000000010000\\
000011110000111111110001111000000011110001110000001000\\
000010001000000011101111000111011010001111101110000100\\
000000000110110011001101111100100001101101011110000010\\
110001001000001111011111110000100000011110111000000001\\
$W(z)=1z^{0}+76z^{24}+51z^{32}$\\
$\# \operatorname{Aut}=16$

\medskip

\noindent
$\codeeff{n}{k}{d}{q}=\codeeff{54}{7}{24}{2}$\\
111111111111111111111110000000000000000000000001000000\\
111111111111000000000001111111111100000000000000100000\\
111100000000111100000001111111100011111110000000010000\\
000011110000111111110001111000000011110001110000001000\\
000010001000000011101111000111011010001111101110000100\\
000000000110110011001101111100100001101101011110000010\\
110010000100001110111110111010000001001110111000000001\\
$W(z)=1z^{0}+76z^{24}+51z^{32}$\\
$\# \operatorname{Aut}=2$

\medskip

\noindent
$\codeeff{n}{k}{d}{q}=\codeeff{54}{7}{24}{2}$\\
111111111111111111111110000000000000000000000001000000\\
111111111111000000000001111111111100000000000000100000\\
111100000000111100000001111111100011111110000000010000\\
000011110000111111110001111000000011110001110000001000\\
000010001000000011101111000111011010001111101110000100\\
000000000110110011001101111100100001101101011110000010\\
110000001100111100111011100110000011000000111110000001\\
$W(z)=1z^{0}+76z^{24}+51z^{32}$\\
$\# \operatorname{Aut}=2$

\medskip

\noindent
$\codeeff{n}{k}{d}{q}=\codeeff{54}{7}{24}{2}$\\
111111111111111111111110000000000000000000000001000000\\
111111111111000000000001111111111100000000000000100000\\
111100000000111100000001111111100011111110000000010000\\
000011110000111111110001111000000011110001110000001000\\
000010001000000011101111000111011010001111101110000100\\
000000000110110011001101111100100001101101011110000010\\
100001111001111011000011110110000100001111001000000001\\
$W(z)=1z^{0}+76z^{24}+51z^{32}$\\
$\# \operatorname{Aut}=2$

\medskip

\noindent
$\codeeff{n}{k}{d}{q}=\codeeff{54}{7}{24}{2}$\\
111111111111111111111110000000000000000000000001000000\\
111111111111000000000001111111111100000000000000100000\\
111100000000111100000001111111100011111110000000010000\\
000011110000111111110001111000000011110001110000001000\\
000010001000000011101111000111011010001111101110000100\\
000000000110110011001101111100100001101101011110000010\\
111010001100111000111000100111110000001000011110000001\\
$W(z)=1z^{0}+76z^{24}+51z^{32}$\\
$\# \operatorname{Aut}=6$

\medskip

\noindent
$\codeeff{n}{k}{d}{q}=\codeeff{54}{7}{24}{2}$\\
111111111111111111111110000000000000000000000001000000\\
111111111111000000000001111111111100000000000000100000\\
111100000000111100000001111111100011111110000000010000\\
000011110000111111110001111000000011110001110000001000\\
000010001000000011101111000111011010001111101110000100\\
000000000110110011001101111100100001101101011110000010\\
111101000100110011101001110010011010001000111000000001\\
$W(z)=1z^{0}+76z^{24}+51z^{32}$\\
$\# \operatorname{Aut}=16$

\medskip

\noindent
$\codeeff{n}{k}{d}{q}=\codeeff{54}{7}{24}{2}$\\
111111111111111111111110000000000000000000000001000000\\
111111111111000000000001111111111100000000000000100000\\
111100000000111100000001111111100011111110000000010000\\
000011110000111111110001111000000011110001110000001000\\
000010001000000011101111000111011010001111101110000100\\
000000000110110011001101111100100001101101011110000010\\
110011100111001100010010100111000001110010111000000001\\
$W(z)=1z^{0}+76z^{24}+51z^{32}$\\
$\# \operatorname{Aut}=16$

\medskip

\noindent
$\codeeff{n}{k}{d}{q}=\codeeff{54}{7}{24}{2}$\\
111111111111111111111110000000000000000000000001000000\\
111111111111000000000001111111111100000000000000100000\\
111100000000111100000001111111100011111110000000010000\\
000011110000111111110001111000000011110001110000001000\\
000010001000000011101111000111011010001111101110000100\\
000000000110110011001101111100100001101101011110000010\\
111101001101101011011111110110111000011111011000000001\\
$W(z)=1z^{0}+76z^{24}+51z^{32}$\\
$\# \operatorname{Aut}=4$

\medskip

\noindent
$\codeeff{n}{k}{d}{q}=\codeeff{54}{7}{24}{2}$\\
111111111111111111111110000000000000000000000001000000\\
111111111111000000000001111111111100000000000000100000\\
111100000000111100000001111111100011111110000000010000\\
000011110000111111110001111000000011110001110000001000\\
000010001000000011101111000111011010001111101110000100\\
000000000110110011001101111100100001101101011110000010\\
111101001101110010111111110110111001001110111000000001\\
$W(z)=1z^{0}+76z^{24}+51z^{32}$\\
$\# \operatorname{Aut}=4$

\medskip

\noindent
$\codeeff{n}{k}{d}{q}=\codeeff{54}{7}{24}{2}$\\
111111111111111111111110000000000000000000000001000000\\
111111111111000000000001111111111100000000000000100000\\
111100000000111100000001111111100011111110000000010000\\
000011110000111111110001111000000011110001110000001000\\
000010001000000011101111000111011010001111101110000100\\
100000000100100011011101111000110011101100011110000010\\
010001001100111000000011000100110011110011111110000001\\
$W(z)=1z^{0}+77z^{24}+49z^{32}+1z^{40}$\\
$\# \operatorname{Aut}=4$

\medskip

\noindent
$\codeeff{n}{k}{d}{q}=\codeeff{54}{7}{24}{2}$\\
111111111111111111111110000000000000000000000001000000\\
111111111111000000000001111111111100000000000000100000\\
111100000000111100000001111111100011111110000000010000\\
000011110000111111110001111000000011110001110000001000\\
000010001000000011101111000111011010001111101110000100\\
100000000100100011011101111000110011101100011110000010\\
100001101000010000101011111110100111010001001110000001\\
$W(z)=1z^{0}+76z^{24}+51z^{32}$\\
$\# \operatorname{Aut}=6$

\medskip

\noindent
$\codeeff{n}{k}{d}{q}=\codeeff{54}{7}{24}{2}$\\
111111111111111111111110000000000000000000000001000000\\
111111111111000000000001111111111100000000000000100000\\
111100000000111100000001111111100011111110000000010000\\
000011110000111111110001111000000011110001110000001000\\
000010001000000011101111000111011010001111101110000100\\
100000000100111011010001100100011101111101001110000010\\
111010000011111011001000100111100100000010101110000001\\
$W(z)=1z^{0}+78z^{24}+48z^{32}+1z^{48}$\\
$\# \operatorname{Aut}=48$

\medskip

\noindent
$\codeeff{n}{k}{d}{q}=\codeeff{54}{7}{24}{2}$\\
111111111111111111111110000000000000000000000001000000\\
111111111111000000000001111111111100000000000000100000\\
111100000000111100000001111111100011111110000000010000\\
000011110000111111110001111000000011110001110000001000\\
000010001000000011101111000111011010001111101110000100\\
100000000100111011010001100100011101111101001110000010\\
111010000110111010101000010111100100001000101110000001\\
$W(z)=1z^{0}+77z^{24}+49z^{32}+1z^{40}$\\
$\# \operatorname{Aut}=12$

\medskip

\noindent
$\codeeff{n}{k}{d}{q}=\codeeff{54}{7}{24}{2}$\\
111111111111111111111110000000000000000000000001000000\\
111111111111000000000001111111111100000000000000100000\\
111100000000111100000001111111100011111110000000010000\\
000011110000111111110001111000000011110001110000001000\\
000010001000000011101111000111011010001111101110000100\\
000001000111110010001001110100000011101111011110000010\\
110011000000110010010000110111111010010000111110000001\\
$W(z)=1z^{0}+78z^{24}+48z^{32}+1z^{48}$\\
$\# \operatorname{Aut}=18$

\medskip

\noindent
$\codeeff{n}{k}{d}{q}=\codeeff{54}{7}{24}{2}$\\
111111111111111111111110000000000000000000000001000000\\
111111111111000000000001111111111100000000000000100000\\
111100000000111100000001111111100011111110000000010000\\
000011110000111111110001111000000011110001110000001000\\
000010001000000011101111000111011010001111101110000100\\
000001000111110010001001110100000011101111011110000010\\
000001111110101001110101110010010110001000001110000001\\
$W(z)=1z^{0}+76z^{24}+51z^{32}$\\
$\# \operatorname{Aut}=6$

\medskip

\noindent
$\codeeff{n}{k}{d}{q}=\codeeff{54}{7}{24}{2}$\\
111111111111111111111110000000000000000000000001000000\\
111111111111000000000001111111111100000000000000100000\\
111100000000111100000001111111100011111110000000010000\\
000011110000111111110001111000000011110001110000001000\\
000010001000000011101111000111011010001111101110000100\\
000001000111110010001001110100000011101111011110000010\\
111100000110110011001101100010111000001100111000000001\\
$W(z)=1z^{0}+76z^{24}+51z^{32}$\\
$\# \operatorname{Aut}=8$

\medskip

\noindent
$\codeeff{n}{k}{d}{q}=\codeeff{54}{7}{24}{2}$\\
111111111111111111111110000000000000000000000001000000\\
111111111111000000000001111111111100000000000000100000\\
111100000000111100000001111111100011111110000000010000\\
000011110000111111110001111000000011110001110000001000\\
000010001000000011101111000111011010001111101110000100\\
110000000110110011000001100000010111111101101110000010\\
001101000111110011101001000110101101001110001000000001\\
$W(z)=1z^{0}+76z^{24}+51z^{32}$\\
$\# \operatorname{Aut}=96$

\medskip

\noindent
$\codeeff{n}{k}{d}{q}=\codeeff{54}{7}{24}{2}$\\
111111111111111111111110000000000000000000000001000000\\
111111111111000000000001111111111100000000000000100000\\
111100000000111100000001111111100011111110000000010000\\
000011110000111111110001111000000011110001110000001000\\
000010001000000011101111000111011010001111101110000100\\
110001100000110000001101100000010111111101101110000010\\
001101110100110010001111000110101101001110001000000001\\
$W(z)=1z^{0}+76z^{24}+51z^{32}$\\
$\# \operatorname{Aut}=480$

\medskip

\noindent
$\codeeff{n}{k}{d}{q}=\codeeff{54}{7}{24}{2}$\\
111111111111111111111110000000000000000000000001000000\\
111111111111000000000001111111111100000000000000100000\\
111100000000111100000001111111100011111110000000010000\\
000011110000111111110001111000000011110001110000001000\\
000010001000000011101111000111011010001111101110000100\\
110011000000110010010000000100111011111101011110000010\\
110010001110001111101000100000100011101110111000000001\\
$W(z)=1z^{0}+76z^{24}+51z^{32}$\\
$\# \operatorname{Aut}=144$

\medskip

\noindent
$\codeeff{n}{k}{d}{q}=\codeeff{54}{7}{24}{2}$\\
111111111111111111111110000000000000000000000001000000\\
111111111111000000000001111111111100000000000000100000\\
111100000000111100000001111111100011111110000000010000\\
000011110000111111110001111000000011110001110000001000\\
000010001000000011101111000111011010001111101110000100\\
110011000000110010010000000100111011111101011110000010\\
101001101100110011110001100010110100000001101110000001\\
$W(z)=1z^{0}+76z^{24}+51z^{32}$\\
$\# \operatorname{Aut}=36$

\medskip

\noindent
$\codeeff{n}{k}{d}{q}=\codeeff{54}{7}{24}{2}$\\
111111111111111111111110000000000000000000000001000000\\
111111111111000000000001111111111100000000000000100000\\
111100000000111100000001111111100011111110000000010000\\
000011110000111111110001111000000011110001110000001000\\
000011000000000011111101100110011011001101101110000100\\
000000001100110011001100011101010111111011101000000010\\
000010100011101000001011111100111011001011011000000001\\
$W(z)=1z^{0}+77z^{24}+49z^{32}+1z^{40}$\\
$\# \operatorname{Aut}=4$

\medskip

\noindent
$\codeeff{n}{k}{d}{q}=\codeeff{54}{7}{24}{2}$\\
111111111111111111111110000000000000000000000001000000\\
111111111111000000000001111111111100000000000000100000\\
111100000000111100000001111111100011111110000000010000\\
000011110000111111110001111000000011110001110000001000\\
000011000000000011111101100110011011001101101110000100\\
000000001100110011001100011101010111111011101000000010\\
110000110000101010100001111101010100000111101110000001\\
$W(z)=1z^{0}+76z^{24}+51z^{32}$\\
$\# \operatorname{Aut}=48$

\medskip

\noindent
$\codeeff{n}{k}{d}{q}=\codeeff{54}{7}{24}{2}$\\
111111111111111111111110000000000000000000000001000000\\
111111111111000000000001111111111100000000000000100000\\
111100000000111100000001111111100011111110000000010000\\
000011110000111111110001111000000011110001110000001000\\
000011000000000011111101100110011011001101101110000100\\
000000001100110011001100011101010111111011101000000010\\
110000111100000000110001010101011011111101010100000001\\
$W(z)=1z^{0}+76z^{24}+51z^{32}$\\
$\# \operatorname{Aut}=48$

\medskip

\noindent
$\codeeff{n}{k}{d}{q}=\codeeff{54}{7}{24}{2}$\\
111111111111111111111110000000000000000000000001000000\\
111111111111000000000001111111111100000000000000100000\\
111100000000111100000001111111100011111110000000010000\\
000011110000111111110001111000000011110001110000001000\\
000011000000000011111101100110011011001101101110000100\\
000000001100110011001100011101010111111011101000000010\\
110000110011001111110001100100101111000110001000000001\\
$W(z)=1z^{0}+79z^{24}+45z^{32}+3z^{40}$\\
$\# \operatorname{Aut}=60$

\medskip

\noindent
$\codeeff{n}{k}{d}{q}=\codeeff{54}{7}{24}{2}$\\
111111111111111111111110000000000000000000000001000000\\
111111111111000000000001111111111100000000000000100000\\
111100000000111100000001111111100011111110000000010000\\
000011110000111111110001111000000011110001110000001000\\
000011000000000011111101100110011011001101101110000100\\
000000001100110011001100011101010111111011101000000010\\
110011000011110011110000011100101111000110001000000001\\
$W(z)=1z^{0}+78z^{24}+48z^{32}+1z^{48}$\\
$\# \operatorname{Aut}=60$

\medskip

\noindent
$\codeeff{n}{k}{d}{q}=\codeeff{54}{7}{24}{2}$\\
111111111111111111111110000000000000000000000001000000\\
111111111111000000000001111111111100000000000000100000\\
111100000000111100000001111111100011111110000000010000\\
000011110000111111110001111000000011110001110000001000\\
000011000000000011111101100110011011001101101110000100\\
000010101100110000001011100101000011111011011110000010\\
000010101111101011000111010011010100001100001110000001\\
$W(z)=1z^{0}+76z^{24}+51z^{32}$\\
$\# \operatorname{Aut}=24$

\medskip

\noindent
$\codeeff{n}{k}{d}{q}=\codeeff{54}{7}{24}{2}$\\
111111111111111111111110000000000000000000000001000000\\
111111111111000000000001111111111100000000000000100000\\
000000000000111100000001111111100011111111111000010000\\
111100000000000011110001111110011011111100000100001000\\
100011100000110011101111110001000010000011111010000100\\
011010011100000010010001001100110101110011111100000010\\
000101111010101011011000001000100001001111111010000001\\
$W(z)=1z^{0}+77z^{24}+49z^{32}+1z^{40}$\\
$\# \operatorname{Aut}=42$

\medskip

\noindent
$\codeeff{n}{k}{d}{q}=\codeeff{54}{7}{24}{2}$\\
111111111111111111111110000000000000000000000001000000\\
111111111111000000000001111111111100000000000000100000\\
111100000000111100000001111000000011111111111000010000\\
000011000000110011110001100111100011111111000100001000\\
100010000000100000001001011111111111110000111110000100\\
000001111000101011111100000111010011001000111010000010\\
000001000111111110001111010000011010100100111010000001\\
$W(z)=1z^{0}+78z^{24}+47z^{32}+2z^{40}$\\
$\# \operatorname{Aut}=240$

\medskip

\noindent
$\codeeff{n}{k}{d}{q}=\codeeff{54}{7}{24}{2}$\\
111111111111111111111110000000000000000000000001000000\\
111111111111000000000001111111111100000000000000100000\\
111100000000111100000001111000000011111111111000010000\\
000011000000110011110001100111100011111111000100001000\\
100010000000100000001001011111111111110000111110000100\\
000001111000101011111100000111010011001000111010000010\\
010000100110111011111001011100000000110100110110000001\\
$W(z)=1z^{0}+77z^{24}+49z^{32}+1z^{40}$\\
$\# \operatorname{Aut}=120$

\medskip

\noindent
$\codeeff{n}{k}{d}{q}=\codeeff{54}{7}{24}{2}$\\
111111111111111111111110000000000000000000000001000000\\
111111111111000000000001111111111100000000000000100000\\
111100000000111100000001111000000011111111111000010000\\
000011000000110011110001100111100011111111000100001000\\
100010000000100000001001011111111111110000111110000100\\
011001100000110011101110011100010011101000110100000010\\
000001011111101011001100011010011111000100100010000001\\
$W(z)=1z^{0}+77z^{24}+49z^{32}+1z^{40}$\\
$\# \operatorname{Aut}=48$

\medskip

\noindent
$\codeeff{n}{k}{d}{q}=\codeeff{54}{7}{24}{2}$\\
111111111111111111111110000000000000000000000001000000\\
111111111111000000000001111111111100000000000000100000\\
111100000000111100000001111000000011111111111000010000\\
000011000000110011110001100111100011111111000100001000\\
000000110000101010001111010111010011111100110010000100\\
000000101000001101101100101111001011100011111100000010\\
111100101110000010011101001100010011010000111100000001\\
$W(z)=1z^{0}+77z^{24}+49z^{32}+1z^{40}$\\
$\# \operatorname{Aut}=24$

\medskip

\noindent
$\codeeff{n}{k}{d}{q}=\codeeff{54}{7}{24}{2}$\\
111111111111111111111110000000000000000000000001000000\\
111111111111000000000001111111111100000000000000100000\\
111100000000111100000001111000000011111111111000010000\\
000011000000110011110001100111100011111111000100001000\\
000000110000101010001111010111010011111100110010000100\\
100000001110000111001001000100001111111011110110000010\\
010010101101000111101101110111000000000110100110000001\\
$W(z)=1z^{0}+76z^{24}+51z^{32}$\\
$\# \operatorname{Aut}=40$

\medskip

\noindent
$\codeeff{n}{k}{d}{q}=\codeeff{54}{7}{24}{2}$\\
111111111111111111111110000000000000000000000001000000\\
111111111111000000000001111111111100000000000000100000\\
111100000000111100000001111000000011111111111000010000\\
000011000000110011110001100111100011111111000100001000\\
000000110000101010001111010111010011111100110010000100\\
100010101110000110000001110100001111111010100110000010\\
010010101101011111001001000111110000000110100110000001\\
$W(z)=1z^{0}+77z^{24}+49z^{32}+1z^{40}$\\
$\# \operatorname{Aut}=5040$

\medskip

\noindent
$\codeeff{n}{k}{d}{q}=\codeeff{54}{7}{24}{2}$\\
111111111111111111111110000000000000000000000001000000\\
111111111111000000000001111111111100000000000000100000\\
111100000000111100000001111000000011111111111000010000\\
000011000000110011110001100111100011111111000100001000\\
000000110000101010001111010111010011111100110010000100\\
110010101100000001001000011110101011111100101100000010\\
101010011010010110101100000100000111111110001010000001\\
$W(z)=1z^{0}+76z^{24}+51z^{32}$\\
$\# \operatorname{Aut}=5760$

\medskip

\noindent
$\codeeff{n}{k}{d}{q}=\codeeff{54}{7}{24}{2}$\\
111111111111111111111110000000000000000000000001000000\\
111111111111000000000001111111111100000000000000100000\\
111100000000111100000001111000000011111111111000010000\\
000011000000110011110001100111100011111111000100001000\\
000000110000101010001111010111010011111100110010000100\\
111010101110001001110001101000100011100010100110000010\\
110110101101010000001111011000010000011110100110000001\\
$W(z)=1z^{0}+78z^{24}+48z^{32}+1z^{48}$\\
$\# \operatorname{Aut}=4320$

\medskip

\noindent
$\codeeff{n}{k}{d}{q}=\codeeff{51}{8}{24}{2}$\\
111111111111111111111110000000000000000000010000000\\
111111111111000000000001111111111100000000001000000\\
111111000000111111000001111110000011111000000100000\\
111000111000111000111001110001110011100110000010000\\
000100100000100110110101101001001111111111100001000\\
100010110110010100101110001101111011000101100000100\\
010010101101001010011111000010111110010111000000010\\
001010101011110011100011100110101001001110100000001\\
$W(z)=1z^{0}+204z^{24}+51z^{32}$\\
$\# \operatorname{Aut}=48960$

\medskip

\noindent
$\codeeff{n}{k}{d}{q}=\codeeff{54}{8}{24}{2}$\\
111111111111111111111110000000000000000000000010000000\\
111111111111000000000001111111111100000000000001000000\\
111111000000111111000001111110000011111000000000100000\\
000000110000110000111101100001111011110111110000010000\\
110000001100001100110001111111100011000111101000001000\\
000000001010111010001101011101000110111111001100000100\\
001100000011000101101000111101110101111110100000000010\\
101100110001101010110100011011100110110000101000000001\\
$W(z)=1z^{0}+156z^{24}+99z^{32}$\\
$\# \operatorname{Aut}=144$

\medskip

\noindent
$\codeeff{n}{k}{d}{q}=\codeeff{55}{8}{24}{2}$\\
1111111111111111111111100000000000000000000000010000000\\
1111111111110000000000011111111111000000000000001000000\\
1111000000001111000000011111111000111111100000000100000\\
0000111100001111111100011110000000111100011100000010000\\
0000000000000000110011011111100110111111011011100001000\\
0000100011101100001000111101110101100011100011100000100\\
1100010011010000001000110001101011011111111010000000010\\
0010000011101011101110101101000010011111000110000000001\\
$W(z)=1z^{0}+141z^{24}+113z^{32}+1z^{40}$\\
$\# \operatorname{Aut}=21$

\medskip

\noindent
$\codeeff{n}{k}{d}{q}=\codeeff{55}{8}{24}{2}$\\
1111111111111111111111100000000000000000000000010000000\\
1111111111110000000000011111111111000000000000001000000\\
1111000000001111000000011111111000111111100000000100000\\
0000111100001111111100011110000000111100011100000010000\\
0000100010000000111011110001110110100011111011100001000\\
0000000001101100110011011111001000011011010111100000100\\
1000111000000010110000101110101111011011110010000000010\\
1000010010011110110011110000011100111100100110000000001\\
$W(z)=1z^{0}+142z^{24}+112z^{32}+1z^{48}$\\
$\# \operatorname{Aut}=21$

\medskip

\noindent
$\codeeff{n}{k}{d}{q}=\codeeff{55}{8}{24}{2}$\\
1111111111111111111111100000000000000000000000010000000\\
1111111111110000000000011111111111000000000000001000000\\
1111000000001111000000011111111000111111100000000100000\\
0000111100001111111100011110000000111100011100000010000\\
0000100010000000111011110001110110100011111011100001000\\
0000000001101100110011011111001000011011010111100000100\\
1000111001011000000000101001111001011011101011100000010\\
0110010010000011110111110110001000000111101110000000001\\
$W(z)=1z^{0}+140z^{24}+115z^{32}$\\
$\# \operatorname{Aut}=40$

\medskip

\noindent
$\codeeff{n}{k}{d}{q}=\codeeff{55}{8}{24}{2}$\\
1111111111111111111111100000000000000000000000010000000\\
1111111111110000000000011111111111000000000000001000000\\
1111000000001111000000011110000000111111111110000100000\\
0000110000001100111100011001111000111111110001000010000\\
0000001100001010100011110101110100111111001100100001000\\
1000000011100001110010010001000011111110111101100000100\\
0100101011010001111011011101110000000001101001100000010\\
0010011110011101110100001110101010000001111010000000001\\
$W(z)=1z^{0}+140z^{24}+115z^{32}$\\
$\# \operatorname{Aut}=960$

\medskip

\noindent
$\codeeff{n}{k}{d}{q}=\codeeff{56}{8}{24}{2}$\\
11111111111111111111111000000000000000000000000010000000\\
11111111111100000000000111111111110000000000000001000000\\
11111100000011111100000111111000001111100000000000100000\\
00000011111111111100000111111000000000011111000000010000\\
00000011000011000011110110000111101111011000111000001000\\
00000000110010100011101101110100011100011110110100000100\\
00000000101001010011011111001110000010110111101100000010\\
00000010000101111010010001100101111000111101011100000001\\
$W(z)=1z^{0}+129z^{24}+122z^{32}+3z^{40}+1z^{48}$\\
$\# \operatorname{Aut}=96$

\medskip

\noindent
$\codeeff{n}{k}{d}{q}=\codeeff{56}{8}{24}{2}$\\
11111111111111111111111000000000000000000000000010000000\\
11111111111100000000000111111111110000000000000001000000\\
11111100000011111100000111111000001111100000000000100000\\
00000011111111111100000111111000000000011111000000010000\\
00000011000011000011110110000111101111011000111000001000\\
00000000110010100011101101110100011100011110110100000100\\
00000000101001010011011111001110000010110111101100000010\\
11000000000000111100110110110110000011011110011100000001\\
$W(z)=1z^{0}+127z^{24}+126z^{32}+1z^{40}+1z^{48}$\\
$\# \operatorname{Aut}=24$

\medskip

\noindent
$\codeeff{n}{k}{d}{q}=\codeeff{56}{8}{24}{2}$\\
11111111111111111111111000000000000000000000000010000000\\
11111111111100000000000111111111110000000000000001000000\\
11111100000011111100000111111000001111100000000000100000\\
00000011111111111100000111111000000000011111000000010000\\
00000011000011000011110110000111101111011000111000001000\\
00000000110010100011101101110100011100011110110100000100\\
00000000101001010011011111001110000010110111101100000010\\
11000000011000111100000100001100011101110111011100000001\\
$W(z)=1z^{0}+127z^{24}+126z^{32}+1z^{40}+1z^{48}$\\
$\# \operatorname{Aut}=168$

\medskip

\noindent
$\codeeff{n}{k}{d}{q}=\codeeff{56}{8}{24}{2}$\\
11111111111111111111111000000000000000000000000010000000\\
11111111111100000000000111111111110000000000000001000000\\
11111100000011111100000111111000001111100000000000100000\\
00000011111111111100000111111000000000011111000000010000\\
00000011000011000011110110000111101111011000111000001000\\
00000000110010100011101101110100011100011110110100000100\\
00000000101001010011011111001110000010110111101100000010\\
11100000001000101100001110001000011100111111011100000001\\
$W(z)=1z^{0}+126z^{24}+128z^{32}+1z^{48}$\\
$\# \operatorname{Aut}=21$

\medskip

\noindent
$\codeeff{n}{k}{d}{q}=\codeeff{56}{8}{24}{2}$\\
11111111111111111111111000000000000000000000000010000000\\
11111111111100000000000111111111110000000000000001000000\\
11111100000011111100000111111000001111100000000000100000\\
00000011111111111100000111111000000000011111000000010000\\
00000011000011000011110110000111101111011000111000001000\\
00000000110010100011101101110100011100011110110100000100\\
10000000001000010011111111000110010011100111101100000010\\
01100000101000111100000100001100011110110111011100000001\\
$W(z)=1z^{0}+126z^{24}+128z^{32}+1z^{48}$\\
$\# \operatorname{Aut}=3$

\medskip

\noindent
$\codeeff{n}{k}{d}{q}=\codeeff{56}{8}{24}{2}$\\
11111111111111111111111000000000000000000000000010000000\\
11111111111100000000000111111111110000000000000001000000\\
11111100000011111100000111111000001111100000000000100000\\
00000011111111111100000111111000000000011111000000010000\\
00000011000011000011110110000111101111011000111000001000\\
00000000110010100011101101110100011100011110110100000100\\
10000000001000010011111111000110010011100111101100000010\\
11100000111001111110000100000100001100100111011100000001\\
$W(z)=1z^{0}+127z^{24}+126z^{32}+1z^{40}+1z^{48}$\\
$\# \operatorname{Aut}=3$

\medskip

\noindent
$\codeeff{n}{k}{d}{q}=\codeeff{56}{8}{24}{2}$\\
11111111111111111111111000000000000000000000000010000000\\
11111111111100000000000111111111110000000000000001000000\\
11111100000011111100000111111000001111100000000000100000\\
00000011111111111100000111111000000000011111000000010000\\
00000011000011000011110110000111101111011000111000001000\\
00000000110010100011101101110100011100011110110100000100\\
10000000001000010011111111000110010011100111101100000010\\
10000011111000011110011101100111001000010000011100000001\\
$W(z)=1z^{0}+129z^{24}+122z^{32}+3z^{40}+1z^{48}$\\
$\# \operatorname{Aut}=16$

\medskip

\noindent
$\codeeff{n}{k}{d}{q}=\codeeff{56}{8}{24}{2}$\\
11111111111111111111111000000000000000000000000010000000\\
11111111111100000000000111111111110000000000000001000000\\
11111100000011111100000111111000001111100000000000100000\\
00000011111111111100000111111000000000011111000000010000\\
00000011000011000011110110000111101111011000111000001000\\
00000000110010100011101101110100011100011110110100000100\\
10000000001000010011111111000110010011100111101100000010\\
11110000111111101000000001100011001100000110011100000001\\
$W(z)=1z^{0}+126z^{24}+128z^{32}+1z^{48}$\\
$\# \operatorname{Aut}=3$

\medskip

\noindent
$\codeeff{n}{k}{d}{q}=\codeeff{56}{8}{24}{2}$\\
11111111111111111111111000000000000000000000000010000000\\
11111111111100000000000111111111110000000000000001000000\\
11111100000011111100000111111000001111100000000000100000\\
00000011111111111100000111111000000000011111000000010000\\
00000011000011000011110110000111101111011000111000001000\\
00000000110010100011101101110100011100011110110100000100\\
10000000001000010011111111000110010011100111101100000010\\
01000011110100011110011110100101011111111111011100000001\\
$W(z)=1z^{0}+127z^{24}+126z^{32}+1z^{40}+1z^{48}$\\
$\# \operatorname{Aut}=4$

\medskip

\noindent
$\codeeff{n}{k}{d}{q}=\codeeff{56}{8}{24}{2}$\\
11111111111111111111111000000000000000000000000010000000\\
11111111111100000000000111111111110000000000000001000000\\
11111100000011111100000111111000001111100000000000100000\\
00000011111111111100000111111000000000011111000000010000\\
00000011000011000011110110000111101111011000111000001000\\
00000000110010100011101101110100011100011110110100000100\\
11000000000001011110001101000111011111000110101100000010\\
00110000000011011010010001111100100011011110011100000001\\
$W(z)=1z^{0}+129z^{24}+122z^{32}+3z^{40}+1z^{48}$\\
$\# \operatorname{Aut}=8$

\medskip

\noindent
$\codeeff{n}{k}{d}{q}=\codeeff{56}{8}{24}{2}$\\
11111111111111111111111000000000000000000000000010000000\\
11111111111100000000000111111111110000000000000001000000\\
11111100000011111100000111111000001111100000000000100000\\
00000011111111111100000111111000000000011111000000010000\\
00000011000011000011110110000111101111011000111000001000\\
00000000110010100011101101110100011100011110110100000100\\
11000000000001011110001101000111011111000110101100000010\\
11000010001011000000011101110000001011111101011100000001\\
$W(z)=1z^{0}+129z^{24}+122z^{32}+3z^{40}+1z^{48}$\\
$\# \operatorname{Aut}=24$

\medskip

\noindent
$\codeeff{n}{k}{d}{q}=\codeeff{56}{8}{24}{2}$\\
11111111111111111111111000000000000000000000000010000000\\
11111111111100000000000111111111110000000000000001000000\\
11111100000011111100000111111000001111100000000000100000\\
00000011111111111100000111111000000000011111000000010000\\
00000011000011000011110110000111101111011000111000001000\\
00000000110010100011101101110100011100011110110100000100\\
11000000101010100000011000000111101110110111101100000010\\
10110000001011010000001001110000010011111111011100000001\\
$W(z)=1z^{0}+128z^{24}+124z^{32}+2z^{40}+1z^{48}$\\
$\# \operatorname{Aut}=3$

\medskip

\noindent
$\codeeff{n}{k}{d}{q}=\codeeff{56}{8}{24}{2}$\\
11111111111111111111111000000000000000000000000010000000\\
11111111111100000000000111111111110000000000000001000000\\
11111100000011111100000111111000001111100000000000100000\\
00000011111111111100000111111000000000011111000000010000\\
00000011000011000011110110000111101111011000111000001000\\
00000000110010100011101101110100011100011110110100000100\\
11100000001000100011001110001000011111111001101100000010\\
10011000001001011000010101111011100000100111011100000001\\
$W(z)=1z^{0}+129z^{24}+122z^{32}+3z^{40}+1z^{48}$\\
$\# \operatorname{Aut}=4$

\medskip

\noindent
$\codeeff{n}{k}{d}{q}=\codeeff{56}{8}{24}{2}$\\
11111111111111111111111000000000000000000000000010000000\\
11111111111100000000000111111111110000000000000001000000\\
11111100000011111100000111111000001111100000000000100000\\
00000011111111111100000111111000000000011111000000010000\\
00000011000011000011110110000111101111011000111000001000\\
00000000110010100011101101110100011100011110110100000100\\
11100000001000100011001110001000011111111001101100000010\\
10011010110000011110110011111000011000000100011100000001\\
$W(z)=1z^{0}+127z^{24}+126z^{32}+1z^{40}+1z^{48}$\\
$\# \operatorname{Aut}=6$

\medskip

\noindent
$\codeeff{n}{k}{d}{q}=\codeeff{56}{8}{24}{2}$\\
11111111111111111111111000000000000000000000000010000000\\
11111111111100000000000111111111110000000000000001000000\\
11111100000011111100000111111000001111100000000000100000\\
00000011111111111100000111111000000000011111000000010000\\
00000011000011000011110110000111101111011000111000001000\\
10000000100000100011111001110110011100111100110100000100\\
01000000100011010011001000001111111010100111101100000010\\
11111111110011011010100111100001100101011110011100000001\\
$W(z)=1z^{0}+127z^{24}+126z^{32}+1z^{40}+1z^{48}$\\
$\# \operatorname{Aut}=4$

\medskip

\noindent
$\codeeff{n}{k}{d}{q}=\codeeff{56}{8}{24}{2}$\\
11111111111111111111111000000000000000000000000010000000\\
11111111111100000000000111111111110000000000000001000000\\
11111100000011111100000111111000001111100000000000100000\\
00000011111111111100000111111000000000011111000000010000\\
00000011000011000011110110000111101111011000111000001000\\
10000000100000100011111001110110011100111100110100000100\\
01000000100010011011100101111001000011111010101100000010\\
00111010000000100011010101111001111110000010011100000001\\
$W(z)=1z^{0}+128z^{24}+124z^{32}+2z^{40}+1z^{48}$\\
$\# \operatorname{Aut}=1$

\medskip

\noindent
$\codeeff{n}{k}{d}{q}=\codeeff{56}{8}{24}{2}$\\
11111111111111111111111000000000000000000000000010000000\\
11111111111100000000000111111111110000000000000001000000\\
11111100000011111100000111111000001111100000000000100000\\
00000011111111111100000111111000000000011111000000010000\\
00000011000011000011110110000111101111011000111000001000\\
10000000100000100011111001110110011100111100110100000100\\
01000000100010011011100101111001000011111010101100000010\\
10111010100001010000000001001111011011111000011100000001\\
$W(z)=1z^{0}+127z^{24}+126z^{32}+1z^{40}+1z^{48}$\\
$\# \operatorname{Aut}=2$

\medskip

\noindent
$\codeeff{n}{k}{d}{q}=\codeeff{56}{8}{24}{2}$\\
11111111111111111111111000000000000000000000000010000000\\
11111111111100000000000111111111110000000000000001000000\\
11111100000011111100000111111000001111100000000000100000\\
00000011111111111100000111111000000000011111000000010000\\
00000011000011000011110110000111101111011000111000001000\\
10000000100000100011111001110110011100111100110100000100\\
01000000100010011011100101111001000011111010101100000010\\
11000011110010011110001101000110110000011000011100000001\\
$W(z)=1z^{0}+129z^{24}+122z^{32}+3z^{40}+1z^{48}$\\
$\# \operatorname{Aut}=16$

\medskip

\noindent
$\codeeff{n}{k}{d}{q}=\codeeff{56}{8}{24}{2}$\\
11111111111111111111111000000000000000000000000010000000\\
11111111111100000000000111111111110000000000000001000000\\
11111100000011111100000111111000001111100000000000100000\\
00000011111111111100000111111000000000011111000000010000\\
00000011000011000011110110000111101111011000111000001000\\
10000000100000100011111001110110011100111100110100000100\\
01000000100010011011100101111001000011111010101100000010\\
11111100110000011011000001100110001010000110011100000001\\
$W(z)=1z^{0}+127z^{24}+126z^{32}+1z^{40}+1z^{48}$\\
$\# \operatorname{Aut}=8$

\medskip

\noindent
$\codeeff{n}{k}{d}{q}=\codeeff{56}{8}{24}{2}$\\
11111111111111111111111000000000000000000000000010000000\\
11111111111100000000000111111111110000000000000001000000\\
11110000000011110000000111111110001111111000000000100000\\
00001111000011111111000111100000001111000111000000010000\\
00000000000000001100110111111001101111110110111000001000\\
00000000110011001010101110010101011111101101000100000100\\
00001110001000000111100001000100111110111101110100000010\\
11100000111010001110110110011010010001011010100000000001\\
$W(z)=1z^{0}+127z^{24}+125z^{32}+3z^{40}$\\
$\# \operatorname{Aut}=9$

\medskip

\noindent
$\codeeff{n}{k}{d}{q}=\codeeff{56}{8}{24}{2}$\\
11111111111111111111111000000000000000000000000010000000\\
11111111111100000000000111111111110000000000000001000000\\
11110000000011110000000111111110001111111000000000100000\\
00001111000011111111000111100000001111000111000000010000\\
00000000000000001100110111111001101111110110111000001000\\
00000000110011001010101110010101011111101101000100000100\\
00001110001000000111100001000100111110111101110100000010\\
11100001110010001100111111010010010000111010100000000001\\
$W(z)=1z^{0}+127z^{24}+125z^{32}+3z^{40}$\\
$\# \operatorname{Aut}=3$

\medskip

\noindent
$\codeeff{n}{k}{d}{q}=\codeeff{56}{8}{24}{2}$\\
11111111111111111111111000000000000000000000000010000000\\
11111111111100000000000111111111110000000000000001000000\\
11110000000011110000000111111110001111111000000000100000\\
00001111000011111111000111100000001111000111000000010000\\
00000000000000001100110111111001101111110110111000001000\\
00000000110011001010101110010101011111101101000100000100\\
10001110000000101001100111001011111100111001000100000010\\
01101000111000111101001111010001101011010000000100000001\\
$W(z)=1z^{0}+126z^{24}+127z^{32}+2z^{40}$\\
$\# \operatorname{Aut}=8$

\medskip

\noindent
$\codeeff{n}{k}{d}{q}=\codeeff{56}{8}{24}{2}$\\
11111111111111111111111000000000000000000000000010000000\\
11111111111100000000000111111111110000000000000001000000\\
11110000000011110000000111111110001111111000000000100000\\
00001111000011111111000111100000001111000111000000010000\\
00000000000000001100110111111001101111110110111000001000\\
00000000110011001010101110010101011111101101000100000100\\
10001110000000101001100111001011111100111001000100000010\\
11101001111000101110000110001000101011110100000100000001\\
$W(z)=1z^{0}+126z^{24}+127z^{32}+2z^{40}$\\
$\# \operatorname{Aut}=8$

\medskip

\noindent
$\codeeff{n}{k}{d}{q}=\codeeff{56}{8}{24}{2}$\\
11111111111111111111111000000000000000000000000010000000\\
11111111111100000000000111111111110000000000000001000000\\
11110000000011110000000111111110001111111000000000100000\\
00001111000011111111000111100000001111000111000000010000\\
00000000000000001100110111111001101111110110111000001000\\
00000000110011001010101110010101011111101101000100000100\\
10000000101111100010000110011011111000101010111000000010\\
01111110110000010110100001011001001111100001100000000001\\
$W(z)=1z^{0}+128z^{24}+123z^{32}+4z^{40}$\\
$\# \operatorname{Aut}=2$

\medskip

\noindent
$\codeeff{n}{k}{d}{q}=\codeeff{56}{8}{24}{2}$\\
11111111111111111111111000000000000000000000000010000000\\
11111111111100000000000111111111110000000000000001000000\\
11110000000011110000000111111110001111111000000000100000\\
00001111000011111111000111100000001111000111000000010000\\
00000000000000001100110111111001101111110110111000001000\\
00000000110011001010101110010101011111101101000100000100\\
10000000101111100010000110011011111000101010111000000010\\
11111110010000110100100001001001100111100011100000000001\\
$W(z)=1z^{0}+128z^{24}+123z^{32}+4z^{40}$\\
$\# \operatorname{Aut}=6$

\medskip

\noindent
$\codeeff{n}{k}{d}{q}=\codeeff{56}{8}{24}{2}$\\
11111111111111111111111000000000000000000000000010000000\\
11111111111100000000000111111111110000000000000001000000\\
11110000000011110000000111111110001111111000000000100000\\
00001111000011111111000111100000001111000111000000010000\\
00000000000000001100110111111001101111110110111000001000\\
00001000100011001011100100011101011110111101000100000100\\
10001000011010000011001011110011111111001100100000000010\\
10000000011101111011011011001001000001000111111000000001\\
$W(z)=1z^{0}+129z^{24}+122z^{32}+3z^{40}+1z^{48}$\\
$\# \operatorname{Aut}=18$

\medskip

\noindent
$\codeeff{n}{k}{d}{q}=\codeeff{56}{8}{24}{2}$\\
11111111111111111111111000000000000000000000000010000000\\
11111111111100000000000111111111110000000000000001000000\\
11110000000011110000000111111110001111111000000000100000\\
00001111000011111111000111100000001111000111000000010000\\
00000000000000001100110111111001101111110110111000001000\\
00001000100011001011100100011101011110111101000100000100\\
10001000011010000011001011110011111111001100100000000010\\
10000110011101001011011000001001001101110111100000000001\\
$W(z)=1z^{0}+129z^{24}+122z^{32}+3z^{40}+1z^{48}$\\
$\# \operatorname{Aut}=6$

\medskip

\noindent
$\codeeff{n}{k}{d}{q}=\codeeff{56}{8}{24}{2}$\\
11111111111111111111111000000000000000000000000010000000\\
11111111111100000000000111111111110000000000000001000000\\
11110000000011110000000111111110001111111000000000100000\\
00001111000011111111000111100000001111000111000000010000\\
00000000000000001100110111111001101111110110111000001000\\
00001000100011001011100100011101011110111101000100000100\\
10001111100010000010000110011010011101101010111000000010\\
01111110011010000101100100011000101111100001100000000001\\
$W(z)=1z^{0}+126z^{24}+127z^{32}+2z^{40}$\\
$\# \operatorname{Aut}=6$

\medskip

\noindent
$\codeeff{n}{k}{d}{q}=\codeeff{56}{8}{24}{2}$\\
11111111111111111111111000000000000000000000000010000000\\
11111111111100000000000111111111110000000000000001000000\\
11110000000011110000000111111110001111111000000000100000\\
00001111000011111111000111100000001111000111000000010000\\
00000000000000001100110111111001101111110110111000001000\\
00001000100011001011100100011101011110111101000100000100\\
10001111100010000010000110011010011101101010111000000010\\
11110110110000111100000110011000000011000110110100000001\\
$W(z)=1z^{0}+125z^{24}+129z^{32}+1z^{40}$\\
$\# \operatorname{Aut}=3$

\medskip

\noindent
$\codeeff{n}{k}{d}{q}=\codeeff{56}{8}{24}{2}$\\
11111111111111111111111000000000000000000000000010000000\\
11111111111100000000000111111111110000000000000001000000\\
11110000000011110000000111111110001111111000000000100000\\
00001111000011111111000111100000001111000111000000010000\\
00000000000000001100110111111001101111110110111000001000\\
00001100110011110000000110011111101100000110110100000100\\
11000010111011111000100101110001100010100101100000000010\\
11111101001011000100100100001001101110010101010000000001\\
$W(z)=1z^{0}+129z^{24}+122z^{32}+3z^{40}+1z^{48}$\\
$\# \operatorname{Aut}=72$

\medskip

\noindent
$\codeeff{n}{k}{d}{q}=\codeeff{56}{8}{24}{2}$\\
11111111111111111111111000000000000000000000000010000000\\
11111111111100000000000111111111110000000000000001000000\\
11110000000011110000000111111110001111111000000000100000\\
00001111000011111111000111100000001111000111000000010000\\
00000000000000001100110111111001101111110110111000001000\\
00001100110011110000000110011111101100000110110100000100\\
11000010111011111000100101110001100010100101100000000010\\
11111101001011001000010010010001101110010101010000000001\\
$W(z)=1z^{0}+130z^{24}+119z^{32}+6z^{40}$\\
$\# \operatorname{Aut}=32$

\medskip

\noindent
$\codeeff{n}{k}{d}{q}=\codeeff{56}{8}{24}{2}$\\
11111111111111111111111000000000000000000000000010000000\\
11111111111100000000000111111111110000000000000001000000\\
11110000000011110000000111111110001111111000000000100000\\
00001111000011111111000111100000001111000111000000010000\\
00000000000000001111000110011111101111110110111000001000\\
00000000110011001100110101011001011111101101000100000100\\
10001110000010001010001001000110011111111001110100000010\\
11101000110001111001001100000001001111010100110100000001\\
$W(z)=1z^{0}+130z^{24}+119z^{32}+6z^{40}$\\
$\# \operatorname{Aut}=32$

\medskip

\noindent
$\codeeff{n}{k}{d}{q}=\codeeff{56}{8}{24}{2}$\\
11111111111111111111111000000000000000000000000010000000\\
11111111111100000000000111111111110000000000000001000000\\
11110000000011110000000111111110001111111000000000100000\\
00001111000011111111000111100000001111000111000000010000\\
00000000000000001111000110011111101111110110111000001000\\
00000000110011001100110101011001011111101101000100000100\\
10001110000010001010001001000110011111111001110100000010\\
10001110001101110110001010011110101100100010000100000001\\
$W(z)=1z^{0}+129z^{24}+122z^{32}+3z^{40}+1z^{48}$\\
$\# \operatorname{Aut}=8$

\medskip

\noindent
$\codeeff{n}{k}{d}{q}=\codeeff{56}{8}{24}{2}$\\
11111111111111111111111000000000000000000000000010000000\\
11111111111100000000000111111111110000000000000001000000\\
11110000000011110000000111111110001111111000000000100000\\
00001111000011111111000111100000001111000111000000010000\\
00001100000000001111110110011001101100110110111000001000\\
10000010000010000000001111000001111111111111110100000100\\
01001000110001110000100100010101001111001111111000000010\\
11110110000011001100011101010011010000110000111000000001\\
$W(z)=1z^{0}+126z^{24}+127z^{32}+2z^{40}$\\
$\# \operatorname{Aut}=6$

\medskip

\noindent
$\codeeff{n}{k}{d}{q}=\codeeff{56}{8}{24}{2}$\\
11111111111111111111111000000000000000000000000010000000\\
11111111111100000000000111111111110000000000000001000000\\
11110000000011110000000111111110001111111000000000100000\\
00001111000011111111000111100000001111000111000000010000\\
00001100000000001111110110011001101100110110111000001000\\
10000010000010000000001111000001111111111111110100000100\\
00000010111011111110100001000101101000100110110100000010\\
01111010110111100001000101000010011000110111000100000001\\
$W(z)=1z^{0}+128z^{24}+123z^{32}+4z^{40}$\\
$\# \operatorname{Aut}=8$

\medskip

\noindent
$\codeeff{n}{k}{d}{q}=\codeeff{56}{8}{24}{2}$\\
11111111111111111111111000000000000000000000000010000000\\
11111111111100000000000111111111110000000000000001000000\\
11110000000011110000000111111110001111111000000000100000\\
00001111000011111111000111100000001111000111000000010000\\
00001100000000001111110110011001101100110110111000001000\\
00000000110011001100110101011001011111101101000100000100\\
00001010000011001010101100100111101111011000110100000010\\
00000110001100000110101111111001010011000011110100000001\\
$W(z)=1z^{0}+127z^{24}+127z^{32}+1z^{56}$\\
$\# \operatorname{Aut}=16$

\medskip

\noindent
$\codeeff{n}{k}{d}{q}=\codeeff{56}{8}{24}{2}$\\
11111111111111111111111000000000000000000000000010000000\\
11111111111100000000000111111111110000000000000001000000\\
11110000000011110000000111111110001111111000000000100000\\
00001111000011111111000111100000001111000111000000010000\\
00001100000000001111110110011001101100110110111000001000\\
00000000110011001100110101011001011111101101000100000100\\
10000010000000101100111101100110101010111111100000000010\\
10001110001111100000111110111000010110100100100000000001\\
$W(z)=1z^{0}+127z^{24}+127z^{32}+1z^{56}$\\
$\# \operatorname{Aut}=6$

\medskip

\noindent
$\codeeff{n}{k}{d}{q}=\codeeff{56}{8}{24}{2}$\\
11111111111111111111111000000000000000000000000010000000\\
11111111111100000000000111111111110000000000000001000000\\
11110000000011110000000111111110001111111000000000100000\\
00001111000011111111000111100000001111000111000000010000\\
00001100000000001111110110011001101100110110111000001000\\
00000000110011001100110101011001011111101101000100000100\\
00001010110010100000101110011110111010101000111000000010\\
00000110111101101100101101000000000110110011111000000001\\
$W(z)=1z^{0}+127z^{24}+127z^{32}+1z^{56}$\\
$\# \operatorname{Aut}=168$

\medskip

\noindent
$\codeeff{n}{k}{d}{q}=\codeeff{56}{8}{24}{2}$\\
11111111111111111111111000000000000000000000000010000000\\
11111111111100000000000111111111110000000000000001000000\\
11110000000011110000000111111110001111111000000000100000\\
00001111000011111111000111100000001111000111000000010000\\
00001100000000001111110110011001101100110110111000001000\\
00001000100011001000111101110001011110111101000100000100\\
10000010000010110110001010011111001001001111111000000010\\
10000110011101111110000001101001110100000100111000000001\\
$W(z)=1z^{0}+127z^{24}+127z^{32}+1z^{56}$\\
$\# \operatorname{Aut}=36$

\medskip

\noindent
$\codeeff{n}{k}{d}{q}=\codeeff{56}{8}{24}{2}$\\
11111111111111111111111000000000000000000000000010000000\\
11111111111100000000000111111111110000000000000001000000\\
11110000000011110000000111111110001111111000000000100000\\
00001111000011111111000111100000001111000111000000010000\\
00001100000000001111110110011001101100110110111000001000\\
00001000100011001000111101110001011110111101000100000100\\
00000011011000000110110011000000111111101101110100000010\\
00000111000111001110111000110110000010100110110100000001\\
$W(z)=1z^{0}+127z^{24}+127z^{32}+1z^{56}$\\
$\# \operatorname{Aut}=32$

\medskip

\noindent
$\codeeff{n}{k}{d}{q}=\codeeff{56}{8}{24}{2}$\\
11111111111111111111111000000000000000000000000010000000\\
11111111111100000000000111111111110000000000000001000000\\
11110000000011110000000111111110001111111000000000100000\\
00001111000011111111000111100000001111000111000000010000\\
00001100000000001111110110011001101100110110111000001000\\
00001000100011001000111101110001011110111101000100000100\\
00000011011000000110110011000000111111101101110100000010\\
11000010100011110111100100010001101110001000110100000001\\
$W(z)=1z^{0}+124z^{24}+131z^{32}$\\
$\# \operatorname{Aut}=16$

\medskip

\noindent
$\codeeff{n}{k}{d}{q}=\codeeff{56}{8}{24}{2}$\\
11111111111111111111111000000000000000000000000010000000\\
11111111111100000000000111111111110000000000000001000000\\
11110000000011110000000111111110001111111000000000100000\\
00001111000011111111000111100000001111000111000000010000\\
00001100000000001111110110011001101100110110111000001000\\
00001000100011001000111101110001011110111101000100000100\\
11100000010010001000101000010111001011110111111000000010\\
11100100001101000000100011100001110110111100111000000001\\
$W(z)=1z^{0}+127z^{24}+127z^{32}+1z^{56}$\\
$\# \operatorname{Aut}=6$

\medskip

\noindent
$\codeeff{n}{k}{d}{q}=\codeeff{56}{8}{24}{2}$\\
11111111111111111111111000000000000000000000000010000000\\
11111111111100000000000111111111110000000000000001000000\\
11110000000011110000000111111110001111111000000000100000\\
00001111000011111111000111100000001111000111000000010000\\
00001100000000001111110110011001101100110110111000001000\\
00001000100011001000111101110001011110111101000100000100\\
10000110010000100111000001101111111011110010100000000010\\
11100110100011101000110000010111110100110010010000000001\\
$W(z)=1z^{0}+126z^{24}+127z^{32}+2z^{40}$\\
$\# \operatorname{Aut}=42$

\medskip

\noindent
$\codeeff{n}{k}{d}{q}=\codeeff{56}{8}{24}{2}$\\
11111111111111111111111000000000000000000000000010000000\\
11111111111100000000000111111111110000000000000001000000\\
11110000000011110000000111111110001111111000000000100000\\
00001111000011111111000111100000001111000111000000010000\\
00001100000000001111110110011001101100110110111000001000\\
00001000100011001000111101110001011110111101000100000100\\
11000010100000000111100101110111100010111000110100000010\\
11000110111111001111101110000001011111110011110100000001\\
$W(z)=1z^{0}+127z^{24}+127z^{32}+1z^{56}$\\
$\# \operatorname{Aut}=8$

\medskip

\noindent
$\codeeff{n}{k}{d}{q}=\codeeff{56}{8}{24}{2}$\\
11111111111111111111111000000000000000000000000010000000\\
11111111111100000000000111111111110000000000000001000000\\
11110000000011110000000111111110001111111000000000100000\\
00001111000011111111000111100000001111000111000000010000\\
00001100000000001111110110011001101100110110111000001000\\
00001000100011001110001101100101101011111000110100000100\\
11110000011000001100000101010011101010101110101100000010\\
11000011000010101111110010101100000101101000110100000001\\
$W(z)=1z^{0}+124z^{24}+131z^{32}$\\
$\# \operatorname{Aut}=12$

\medskip

\noindent
$\codeeff{n}{k}{d}{q}=\codeeff{56}{8}{24}{2}$\\
11111111111111111111111000000000000000000000000010000000\\
11111111111100000000000111111111110000000000000001000000\\
11110000000011110000000111111110001111111000000000100000\\
00001111000011111111000111100000001111000111000000010000\\
00001100000000001111110110011001101100110110111000001000\\
00001010000011001100101101000111101111101000110100000100\\
11000000110000111100000111100111101100000110101100000010\\
11000110001111110000011100100111100000011000011100000001\\
$W(z)=1z^{0}+127z^{24}+127z^{32}+1z^{56}$\\
$\# \operatorname{Aut}=2304$

\medskip

\noindent
$\codeeff{n}{k}{d}{q}=\codeeff{56}{8}{24}{2}$\\
11111111111111111111111000000000000000000000000010000000\\
11111111111100000000000111111111110000000000000001000000\\
11110000000011110000000111111110001111111000000000100000\\
00001111000011111111000111100000001111000111000000010000\\
00001100000000001111110110011001101100110110111000001000\\
00001010000011001100101101000111101111101000110100000100\\
11000000110011000011000111100111101100000110101100000010\\
11000110001100001111011100100111100000011000011100000001\\
$W(z)=1z^{0}+127z^{24}+127z^{32}+1z^{56}$\\
$\# \operatorname{Aut}=1536$

\medskip

\noindent
$\codeeff{n}{k}{d}{q}=\codeeff{56}{8}{24}{2}$\\
11111111111111111111111000000000000000000000000010000000\\
11111111111100000000000111111111110000000000000001000000\\
11110000000011110000000111111110001111111000000000100000\\
00001111000011111111000111100000001111000111000000010000\\
00001100000000001111110110011001101100110110111000001000\\
00000000110011111100000101000111010011101101110100000100\\
11000000000011001111000100100110111100011011101100000010\\
11001100001100110011110111100110001100000000011100000001\\
$W(z)=1z^{0}+127z^{24}+127z^{32}+1z^{56}$\\
$\# \operatorname{Aut}=192$

\medskip

\noindent
$\codeeff{n}{k}{d}{q}=\codeeff{56}{8}{24}{2}$\\
11111111111111111111111000000000000000000000000010000000\\
11111111111100000000000111111111110000000000000001000000\\
11110000000011110000000111111110001111111000000000100000\\
00001111000011111111000111100000001111000111000000010000\\
00001100000000001111110110011001101100110110111000001000\\
00000000110011111100000101000111010011101101110100000100\\
11001111000011000000000100100110110011011011101100000010\\
11000011001100111100110111100110000011000000011100000001\\
$W(z)=1z^{0}+127z^{24}+127z^{32}+1z^{56}$\\
$\# \operatorname{Aut}=192$

\medskip

\noindent
$\codeeff{n}{k}{d}{q}=\codeeff{56}{8}{24}{2}$\\
11111111111111111111111000000000000000000000000010000000\\
11111111111100000000000111111111110000000000000001000000\\
11110000000011110000000111111110001111111000000000100000\\
00001111000011111111000111100000001111000111000000010000\\
00001100000000001111110110011001101100110110111000001000\\
11001010000011000000101101000001101111101110110100000100\\
00110011110011001100110110000000001111110000101100000010\\
11110101001100001100101101000110000011101000011100000001\\
$W(z)=1z^{0}+127z^{24}+127z^{32}+1z^{56}$\\
$\# \operatorname{Aut}=7680$

\medskip

\noindent
$\codeeff{n}{k}{d}{q}=\codeeff{56}{8}{24}{2}$\\
11111111111111111111111000000000000000000000000010000000\\
11111111111100000000000111111111110000000000000001000000\\
00000000000011110000000111111110001111111111100000100000\\
11110000000000001111000111111001101111110000010000010000\\
00001111000000001100110111111001101100001111001000001000\\
10001000110011000010100000000111011011101110111100000100\\
01001110101010101000111000000001011110011001011100000010\\
10100101011001101001101111111000000101000100100100000001\\
$W(z)=1z^{0}+130z^{24}+119z^{32}+6z^{40}$\\
$\# \operatorname{Aut}=384$

\medskip

\noindent
$\codeeff{n}{k}{d}{q}=\codeeff{56}{8}{24}{2}$\\
11111111111111111111111000000000000000000000000010000000\\
11111111111100000000000111111111110000000000000001000000\\
11110000000011110000000111100000001111111111100000100000\\
00001100000011001111000110011110001111111100010000010000\\
00000011000000111100110110011001101111110011001000001000\\
00000000110010101000111101010110011111101010100100000100\\
11000000000010101110100110010101010000011111111100000010\\
11000000001100000101101100100100101111111010011100000001\\
$W(z)=1z^{0}+127z^{24}+127z^{32}+1z^{56}$\\
$\# \operatorname{Aut}=240$

\medskip

\noindent
$\codeeff{n}{k}{d}{q}=\codeeff{56}{8}{24}{2}$\\
11111111111111111111111000000000000000000000000010000000\\
11111111111100000000000111111111110000000000000001000000\\
11110000000011110000000111100000001111111111100000100000\\
00001100000011001111000110011110001111111100010000010000\\
00000011000000111100110110011001101111110011001000001000\\
00000000110010101000111101010110011111101010100100000100\\
11000000000010101110100110010101010000011111111100000010\\
10101010101011110010010101011110000000011100100100000001\\
$W(z)=1z^{0}+124z^{24}+131z^{32}$\\
$\# \operatorname{Aut}=240$

\medskip

\noindent
$\codeeff{n}{k}{d}{q}=\codeeff{56}{8}{24}{2}$\\
11111111111111111111111000000000000000000000000010000000\\
11111111111100000000000111111111110000000000000001000000\\
11110000000011110000000111100000001111111111100000100000\\
00001100000011001111000110011110001111111100010000010000\\
00000011000000111100110110011001101111110011001000001000\\
00000000110010101000111101010110011111101010100100000100\\
00001110100000001010110001111111011100010101101000000010\\
00001110101110100001111011001110100011110000001000000001\\
$W(z)=1z^{0}+127z^{24}+127z^{32}+1z^{56}$\\
$\# \operatorname{Aut}=42$

\medskip

\noindent
$\codeeff{n}{k}{d}{q}=\codeeff{56}{8}{24}{2}$\\
11111111111111111111111000000000000000000000000010000000\\
11111111111100000000000111111111110000000000000001000000\\
11110000000011110000000111100000001111111111100000100000\\
00001100000011001111000110011110001111111100010000010000\\
00000011000000111100110110011001101111110011001000001000\\
11000000110000001010101101000101001111111010011100000100\\
10101010101010100101101000010000011111110110000100000010\\
01101010100110101100110010110010110000000011011100000001\\
$W(z)=1z^{0}+127z^{24}+127z^{32}+1z^{56}$\\
$\# \operatorname{Aut}=5760$

\medskip

\noindent
$\codeeff{n}{k}{d}{q}=\codeeff{56}{9}{24}{2}$\\
11111111111111111111111000000000000000000000000100000000\\
11111111111100000000000111111111110000000000000010000000\\
11110000000011110000000111100000001111111111100001000000\\
00001100000011001111000110011110001111111100010000100000\\
00000011000010101000111101011101001111110011001000010000\\
10000000111000011100100100010000111111101111011000001000\\
01001010110100011110110111011100000000011010011000000100\\
00100111100111011101000011101010100000011110100000000010\\
11101101010101001000101100000101011111100100100000000001\\
$W(z)=1z^{0}+255z^{24}+255z^{32}+1z^{56}$\\
$\# \operatorname{Aut}=48960$

\medskip

\noindent
$\codeeff{n}{k}{d}{q}=\codeeff{56}{9}{24}{2}$\\
11111111111111111111111000000000000000000000000100000000\\
11111111111100000000000111111111110000000000000010000000\\
11110000000011110000000111111110001111111000000001000000\\
00001111000011111111000111100000001111000111000000100000\\
00001000100000001110111100011101101000111110111000010000\\
00000000011011001100110111110010000110110101111000001000\\
10001110010110000000001010011110010110111010111000000100\\
01100100100000111101111101100010000001111011100000000010\\
11100010110001110000111100001101110110001010100000000001\\
$W(z)=1z^{0}+255z^{24}+255z^{32}+1z^{56}$\\  
$\# \operatorname{Aut}=1440$

\end{document}